\documentclass{article}

\usepackage{amsmath,amsfonts,amssymb,bbm,amsthm, bm}
\usepackage{xcolor}
\usepackage{siunitx}
\usepackage[letterpaper, margin=1in]{geometry}
\usepackage{thm-restate}

\usepackage{hyperref}
\usepackage[square,numbers]{natbib}

\usepackage{hyperref}
\usepackage[capitalize]{cleveref}
\usepackage{mathtools}
\usepackage{float}
\usepackage{graphicx}
\usepackage{subcaption}
\usepackage{enumitem}
\usepackage{caption}
\usepackage{amsmath}
\usepackage{bbm}
\usepackage{mathrsfs}
\usepackage[at]{easylist}

\usepackage[ruled,linesnumberedhidden, noend]{algorithm2e}
\newenvironment{mechanism}[1][htb]
  {
   \begin{algorithm}[#1]%
  }{\end{algorithm}}
\SetKw{Continue}{continue}
\SetKw{Break}{break}
\usepackage{hhline}

\usepackage{booktabs} 

\newtheorem{theorem}{Theorem}[section]
\newtheorem{lemma}[theorem]{Lemma}

\newtheorem{prop}[theorem]{Proposition}

\newtheorem{definition}{Definition}
\newtheorem{example}{Example}[section]

\newcommand{\E}[0]{{\mathbb{E} }}

\usepackage{color-edits}
\addauthor[Grant]{gs}{red}
\addauthor[Yichi]{yichi}{blue}

\newcommand{\ProbRev}{\ensuremath{P_\text{rev}}\xspace}
\newcommand{\ProbAcc}{\ensuremath{P_\text{acc}}\xspace}

\begin{document}

\title{Eliciting Honest Information From Authors Using Sequential Review}

\author{
  Yichi Zhang\thanks{School of Information, University of Michigan, \texttt{yichiz@umich.edu}}\and Grant Schoenebeck\thanks{School of Information, University of Michigan, \texttt{schoeneb@umich.edu}}\and Weijie Su\thanks{Department of Computer and Information Science, University of Pennsylvania, \texttt{suw@wharton.upenn.edu}}
}
\date{}
\maketitle

\begin{abstract}
    In the setting of conference peer review, the conference aims to accept high-quality papers and reject low-quality papers based on noisy review scores. A recent work proposes the isotonic mechanism, which can elicit the ranking of paper qualities from an author with multiple submissions to help improve the conference's decisions. However, the isotonic mechanism relies on the assumption that the author's utility is both an increasing and a convex function with respect to the review score, which is often violated in peer review settings (e.g.~when authors aim to maximize the number of accepted papers). In this paper, we propose a sequential review mechanism that can truthfully elicit the ranking information from authors while only assuming the agent's utility is increasing with respect to the true quality of her accepted papers.  
    The key idea is to review the papers of an author in a sequence based on the provided ranking and conditioning the review of the next paper on the review scores of the previous papers.
    Advantages of the sequential review mechanism include 1) eliciting truthful ranking information in a more realistic setting than prior work; 2) improving the quality of accepted papers, reducing the reviewing workload and increasing the average quality of papers being reviewed; 3) incentivizing authors to write fewer papers of higher quality.
\end{abstract}

\section{Introduction}


Peer review, the process of evaluating scientific research by volunteered experts, undergirds the success of a conference by ensuring the accepted papers are of high quality.  However, the reliability of peer review (especially for large computer science conferences) has raised significant concerns. In a NeurIPS experiment conducted in 2014 \cite{lawrence2014nips}, it was shown that within the set of papers recommended for acceptance by two independent committees, the disagreement rate was as high as 50\%. This result was further confirmed in the repeated experiment conducted in 2021 \cite{cortes2021inconsistency}. Even worse, the rapid growth of the reviewing workload and the shortage of qualified reviewers, have posed unprecedented challenges to our review system \cite{sculley2018avoiding, Shah2019PrincipledMT}.  

This leads to the dilemma of conference peer review: the conference's objective of accepting only high-quality papers (from a large set of submissions) clashes with the shortage of reliable peer reviews upon which the conference must base its decisions. To mitigate this issue, we introduce a novel review mechanism called the \emph{sequential review mechanism} that can 1) solicit high-quality information from authors to assist the acceptance/rejection decisions, 2) reduce the reviewing workload and 3) incentivize authors to write high-quality papers. 

The main challenge is how to elicit useful information from the authors who have conflicting interests with the conference. For instance, authors may desire to have more publications, leading them to seek acceptance for more papers, regardless of the papers' quality. In this case, although authors possess the best signals of their own papers' quality compared with any reviewer, they may prefer not to disclose this information truthfully to the conference. For example, while being asked to report the true quality of their papers, authors may be inclined to inflate scores in order to increase the chance of acceptance. 

Fortunately, positive results exist. \citet{su2021you} shows that it is possible to elicit truthful rankings of paper quality from an author with multiple submissions. The main idea of the proposed \emph{isotonic mechanism} is to shift the noisy review scores by running an isotonic regression based on the author's reported ranking. It is shown that reporting the ranking of papers truthfully is the best response for an author.  Unfortunately, this result only holds when the agent's utility for each paper is an \textbf{increasing and convex} function of the review score and additive across papers. 
However, the assumption of convex utility is strong and likely violated in the setting of conference peer review. For example, this assumption is violated when an author aims to maximize the number of her accepted papers.  Moreover, in this case, the isotonic mechanism can be gamed in a rather straightforward manner.  Suppose such an author has several borderline papers and one outstanding paper to submit. Under the isotonic mechanism, her best response is nonetheless to rank the outstanding paper at the bottom such that the review scores of all borderline papers will be shifted up after the isotonic regression, which will almost certainly lead to the acceptance of all papers.\footnote{For an introduction of the isotonic mechanism and a careful illustration of this example, please refer \cref{app:isotonic}.}

In this work, we build on the idea of eliciting the author's ranking information from the previous work and primarily focus on addressing the incentive issue discussed above.
We propose the sequential review mechanism in a natural conference review model. Our method can elicit the true ranking information as long as the author's utility is additive in terms of the rewards of all \textbf{accepted} papers, where the reward is a \textbf{non-decreasing} function of the accepted papers' qualities.  
The sequential review mechanism works by reviewing an author's submissions in sequence. In particular, papers with higher reported rankings are reviewed with priority, while papers with lower reported rankings will be conditionally reviewed depending on the review scores of the higher-ranked papers from the same author. If the review process terminates, e.g.~due to a notably low review score of a paper, any remaining unreviewed papers will be rejected without further assessment. Intuitively, under the sequential review mechanism, any misreporting of the true ranking of the papers will result in an earlier termination of the review process, which penalizes dishonest behaviors.\footnote{One may be concerned that the sequential review mechanism will result in a significant delay in the review process. However, note that the mechanism works exactly the same if all papers are simultaneously reviewed or reviewed in batches, as long as the acceptance/rejection decisions are made in sequence. See more in \cref{sec:discussion}.}

\subsection{Contributions and Results}

Our main contribution is a framework for designing theoretically robust mechanisms, with the potential to improve conference peer review in practice. The proposed sequential review mechanism not only addresses a key incentive issue that plagues prior work but also exhibits many additional appealing properties. 

\paragraph{Truthful sequential review mechanisms.}
Under the sequential review mechanism framework, we first identify a sufficient constraint that ensures a sequential review mechanism to be truthful (i.e.~reporting the true ranking of the paper quality is the best response for any author). While not necessary, this constraint provides a large space of truthful sequential review mechanisms. To show the effectiveness of our framework, we introduce two practical mechanisms as examples: the \emph{memoryless coin-flip mechanism} that reviews the $i+1$th ranked paper with a probability determined by the review score of the $i$th ranked paper; and the \emph{credit pool mechanism}, which counts the cumulative review scores (positive or negative) of the reviewed papers and terminates the review process when the ``credit pool'' is empty.

\paragraph{Conference utility and review burden.}
The sequential review mechanism utilizes the authors' information to prioritize the review of high-quality papers. Therefore, it can improve the conference's utility (as low-quality papers are less likely to be accepted), while reducing the review burden by reallocating more review resources to papers deemed likely to be of higher quality.
To evaluate the performance of the sequential review mechanism, we use the \emph{parallel review mechanism} as the baseline which unconditionally reviews all papers. We further use the isotonic mechanism with oracle access to the true ranking information as an unachievable upper bound.\footnote{The isotonic mechanism is not truthful in our setting but we nonetheless provide oracle access to the true ranking to it.}

Our simulation results suggest that compared with the baseline, the sequential review mechanism can improve the conference utility towards the upper bound by over $40\%$ when the author submits more than three papers. This effect is even more significant when 1) each author has more papers, 2) papers are more likely to be of low quality and 3) reviewers are more noisy. 
Moreover, we empirically investigate the number of reviews that a sequential review mechanism can save while achieving the same conference utility as the parallel review mechanism. We employ the ICLR OpenReview datasets spanning recent years and develop a more realistic review model. Our results indicate that about 20\% of the review burden can be saved when utilizing the sequential review mechanism. Furthermore, this number will increase over time if the trend of a  growing number of submissions per author continues.

\paragraph{Endogenous paper quality.}
In the setting where authors can choose the effort they exert on each of their papers, we show that compared with the parallel review mechanism, the sequential review mechanism always provides a stronger incentive for writing (fewer) papers of higher quality instead of (more) papers of low-quality. This is because the sequential review mechanism decreases the marginal return of producing lower-quality papers by penalizing bottom-ranked papers with lower probabilities of being reviewed. We view this property particularly valuable, especially in light of the prevailing trend where authors submit an increasingly large number of papers to conferences, sometimes disregarding their inherent quality. 

\section{Related Works}

Other than the discussions on the isotonic mechanism \cite{su2021you, wu2023isotonic, yan2023isotonic}, several attempts exist with the goal of improving the peer review system with a focus on dealing with strategic interactions between conferences and authors. In a setting where authors can strategically decide the venues to submit their papers, \citet{10.1145/3490486.3538235} model the paper (re)submission process as a Stackelberg game between the author and the conference. They focus on how to design the review mechanism to achieve the Pareto optimal tradeoff between the conference quality and the review burden. \citet{srinivasan2021auctions} propose the idea of using the VCG mechanism to elicit bids from authors and using peer prediction mechanisms to evaluate reviews and reward the reviewers (with virtual money). Thus, agents are motivated to provide high-quality reviews so as to raise enough funding to bid for review slots. In dealing with the malicious bidding problem, a stream of literature focuses on designing and optimizing the paper-reviewer assignment mechanism \cite{aziz2019strategyproof, jecmen2020mitigating, dhull2022strategyproofing, xu2018strategyproof}. The goal here is to achieve strategyproofness such that the outcomes of any reviewer's own submissions are (approximately) independent of the reviews they provide for other submissions.

Our work is also related to the impartial peer selection problem, where self-interested agents assess one another in such a way that none of them has an incentive to misrepresent their evaluation. A notable application of this problem is the US National Science Foundation (NSF) experiment, each PI was asked to rank 7 proposals from other PIs \citep{naghizadeh2013incentives}. To incentivize careful reviews, the chance for a PI to be funded is tied to the quality of their reviews. The primary goal of the literature on the peer selection problem is to improve the accuracy of assessments while guaranteeing strategyproofness \cite{naghizadeh2013incentives, aziz2019strategyproof, de2008impartial, 6920096}. However, these investigations differ from our problem in that our focus is on eliciting evaluations of multiple items held by a single agent directly from the agent itself, rather than relying on evaluations from other agents.

Additionally, there exists a considerable body of literature that aims to improve peer review from the reviewer's perspective.
This includes investigations into single versus double blind reviewing \cite{blank1991effects, snodgrass2006single, bazi2020peer}, assignment versus bidding \cite{asi.22747, meir2021market}, review scale and miscalibration \cite{radiology.178.3.1994394, wang20182, spalvieri2014weighting}, reviewer bias \cite{lee2013bias, haffar2019peer, lane2022conservatism}, and dishonest behaviors \cite{cohen2016organised, fanelli2009many, littman2021collusion}.
A recent survey by \citet{shah2022challenges} provides additional contexts and perspectives on the problems of peer review.

\section{Model}
\label{sec:model}

We view the peer review mechanism as an individual contract. That is, each paper is reviewed and evaluated independently based on its review scores.\footnote{An example of a review policy that is not an individual contract is when the conference maintains a fixed acceptance rate, which results in a contest setting where authors have to compete with each other for acceptance.  However, in such a case, as the number of authors and papers becomes large, the importance of interactions between individual authors will decrease and therefore limit our model.} Therefore, while reasoning about an agent's best response (\cref{sec:seq_mechanism} and \ref{sec:endogeneous}), it is sufficient to assume that there is only one agent with $n$ papers. We eventually investigate the optimization of a review mechanism, where we assume authors are drawn from a distribution (see \cref{subsec:real_data}).

Throughout the paper, we will use $[n]$ to denote the set $\{1,2,\ldots,n\}$. 
Suppose an author has $n$ submissions indexed by $i\in [n]$, each with a quality of $q_i\in \mathbb{R}$.  
The author observes a noisy signal $s_i = q_i + \xi_i$ for each of her papers where $\xi_i$ are i.i.d.~sampled from some distribution. This information is private to the author.\footnote{Our main results in \cref{sec:seq_mechanism} only assume the author has a noisy signal about his papers' qualities. However, in  \cref{sec:opt_mechanism} and \cref{sec:endogeneous}, we restrict our attention to the setting where authors observe perfect information, i.e.~$s_i = q_i$.}

For each paper $i$, the conference decides whether to accept or reject based on its review score $r_i=q_i + \epsilon_i$, where $\epsilon_i$ are i.i.d.~sampled from some distribution. 
In the theory sections, we will further simplify by assuming each paper is assigned with one reviewer; however,  we show how to apply our method to the multi-reviewer setting in \cref{subsec:real_data}.
The conference commits to an acceptance policy such that a paper with review score $r$ (if it is reviewed) is accepted with probability $\ProbAcc(r)$. For example, for a threshold acceptance policy, $\ProbAcc(r) = 1$ if $r\ge \tau_{acc}$ and 0 otherwise. We assume that the utility of the conference is the sum of the accepted papers' quality, i.e.~$U_c(\mathcal{M}) = \sum_{i\in [n]} q_i\cdot \mathbbm{1}[\text{paper $i$ is accepted under mechanism $\mathcal{M}$}]$. That is, the quality of a paper can be viewed as its desirability for acceptance by the conference.

In addition to soliciting review scores, the conference can solicit a ranking of paper qualities from the author. The author first ranks her papers based on her review scores $s_i$ such that, without loss of generality, $s_1 \ge s_2 \ge \cdots \ge s_n$. We name each paper by the author's true ranking, i.e.~paper $i$ is the paper that the author's signal is $s_i$.  
The author reports a permutation $\pi$ on the indexes $1, 2, \ldots, n$ of her papers, where $\pi(i)$ is the rank of paper $i$ after the permutation. The truthful report is the original ranking, i.e.~$\pi^*(i)=i$. 
We primarily focus on the incentive issue in the single-author setting, i.e.~each paper has only one author. However, we note an easy solution in the multi-author setting: assign each paper to one of its authors and solicit the ranking information only from the selected author (see \cref{sec:discussion} for details).

We assume that the author's utility is the sum of the rewards of her accepted papers: each paper's reward is zero if rejected and $u_a(q_i)$ if accepted, where $u_a$ is a non-negative and non-decreasing reward function for the true quality of the paper. For example, if $u_a(q)=1$ for any $q$, the author's goal is to maximize the expected number of accepted papers. Note that the reward is a function of the true quality rather than the author's signal because the value of an accepted paper (e.g.~its citations) largely depends on how useful others find it.
The author strategically reports a ranking $\pi$ to maximize its expected utility. 
We use $U_a(\pi)$ to denote the expected author utility under the permutation $\pi$, where the randomness is with respect to the review noise and the mechanism. 

The main question studied in this paper (\cref{sec:seq_mechanism}) is how to design a \emph{truthful} review mechanism such that reporting $\pi^*$ is the author's best response, i.e.~$U_a(\pi^*)\ge U_a(\pi)$ for any $\pi$. Then, in \cref{sec:opt_mechanism}, we study how to (empirically) optimize the conference utility conditioned on truthfulness.
Finally, in \cref{sec:endogeneous}, we consider a variant of the general model by additionally considering that the author can choose the quality of the papers that she writes. 

\section{Truthful Sequential Review Mechanisms}
\label{sec:seq_mechanism}

This section presents a framework for designing truthful review mechanisms. In particular, we introduce the sequential review mechanism framework and show a sufficient condition for a sequential review mechanism to be truthful. We further provide two concrete and practical truthful sequential review mechanisms under this framework as examples.

\subsection{The Sequential Review Mechanism Framework}

We first introduce the \emph{naive sequential review mechanism} as an illustrative example of the more general sequential review mechanism framework.  

\begin{definition}
    Given an author with $n$ papers and a ranking of these papers, the naive sequential review mechanism reviews one paper at a time based on the order of the reported ranking. The first paper is always reviewed.  However, for $i$ ranging from $2$ to $n$, the paper ranked in the $i$th place is reviewed if and only if the paper ranked in $i-1$st place is accepted.
\end{definition}

Intuitively, the naive sequential review mechanism incentivizes truth-telling because any manipulation of the true ranking will more likely result in an early stop of the review process which harms the utility of the author. However, the naive sequential review mechanism can be too stringent in reality, especially when authors are likely to produce good papers. To address this, we generalize this idea and present the sequential review mechanism framework. This framework offers a diverse range of mechanisms that can be fine-tuned to optimize performance in various circumstances.

At a high level, the idea is to condition the review of lower-ranked papers on the acceptance (and thus the review score) of the higher-ranked papers. If the mechanism decides not to review the paper in round $i$, any paper in round $j>i$ will be rejected without review. We then say that the mechanism terminates in round $i$. Now, we formally introduce the sequential review mechanism framework.

\begin{definition}\label{def:framework}
    A sequential review mechanism $\mathcal{M}_s=(\ProbAcc, \ProbRev, \bm{\mu})$ has three components:
\begin{itemize}
    \item An \emph{acceptance policy} $\ProbAcc$ that maps from a review score to a probability of accepting the corresponding paper.
    \item A \emph{review policy} $\ProbRev$ that maps from a review state in round $i$ to a probability of reviewing the paper in round $i$, for $i\in [n]$.
    \item A \emph{state transition mapping} $\mu_i$ that maps from a review state and a review score of the paper in round $i$ to a distribution of states in round $i+1$, for $i\in [n-1]$.
\end{itemize}
\end{definition}

\begin{figure*}[tb]
     \centering
     \begin{subfigure}[b]{1\textwidth}
         \centering
         \includegraphics[width=\textwidth]{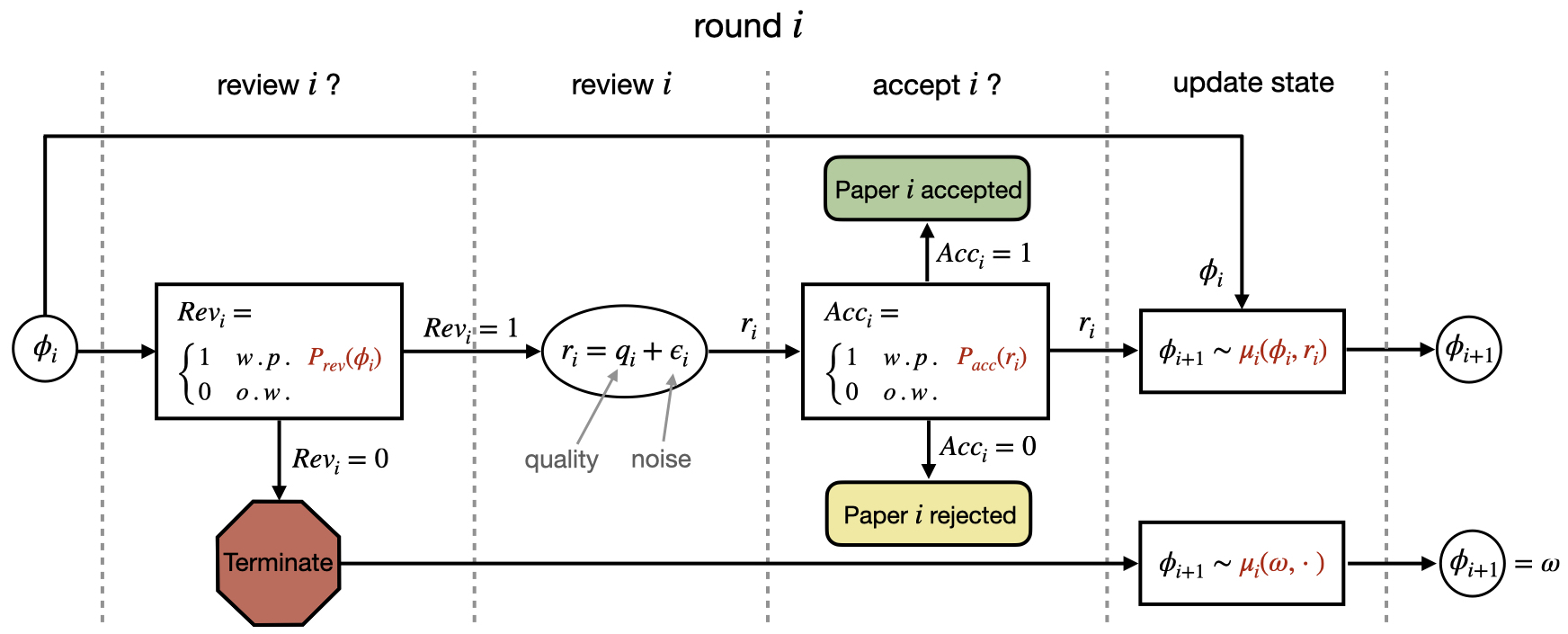}
     \end{subfigure}
     \hfill
     \caption{The sequential review mechanism in round $i$.}\label{fig:round_i}
\end{figure*}

\vspace{-2mm}
The sequential review mechanism in round $i$ is illustrated in \cref{fig:round_i}. 
In the above definition, a \emph{review state} in round $i$, denoted as $\phi_i\in \Phi_i$, is sufficient to determine the probability that the paper in round $i$ will be reviewed, where $\Phi_i$ is the space of review states in round $i$. The review states are weakly ordered such that the author always prefers to be in a higher-ordered state. That is, for each pair of states in round $i$, one of which must have a (weakly) higher order than the other, denoted as $\phi'_i \succeq \phi_i$ for every $\phi'_i, \phi_i \in \Phi_i$ and every $i\in [n]$. We use $\phi'_i\succ \phi_i$ to represent the strict ordering. Furthermore, if the author is indifferent between two states, we say $\phi'_i \sim \phi_i$. 
When comparing two vectors of states, we use the same notations to indicate term-wise preference. For example, if $\bm{\phi}, \bm{\phi}'\in\Phi^m$, $\bm{\phi}' \succeq \bm{\phi}$ implies that $\phi'_j \succeq \phi_j$ for any $1\le j\le m$.

Specially, we use $\omega$ to denote the termination state such that $\ProbRev(\omega) = 0$ and $\mu(\omega, \cdot) = \omega$ with probability of $1$. Note that for any round $i$ and any review state $\phi\in \Phi_i\backslash \{\omega\}$, we have $\phi \succ \omega$.

Taking the naive sequential review mechanism as an example, the review state is $\phi_i = \psi$ if the paper in round $i$ is accepted and $\omega$ otherwise. In round $1$, $\phi_1 = \psi$. Then, the review policy is $\ProbRev^{std}(\psi)=1$ and $\ProbRev^{std}(\omega)=0$ for any round. The state transition mapping of the naive mechanism is that $\mu^{std}(\psi,r) = \psi$ with probability $\ProbAcc(r)$, $\mu^{std}(\psi,r) =\omega$ with probability $1-\ProbAcc(r)$, and $\mu^{std}(\omega,\cdot) = \omega$ with probability $1$.

Furthermore, note that by our definition, the acceptance policy is assumed to be memoryless, where the acceptance of a paper only depends on its own review score. In other words, conditioned on a paper being reviewed, the same acceptance policy is applied. However, the review policy can have memory such that review scores from previous rounds may affect the distribution of the review state in the current round through a Markov chain.


\subsection{A Sufficient Condition For Truthfulness}
\label{sec:sufficient_cond}

Now, we investigate what properties of the acceptance policy, the review policy, and the state transition mapping are sufficient for a sequential review mechanism to be truthful. First of all, we need both policies to be monotone which reward higher review scores and punishes lower review scores.

\begin{definition}
    We say an acceptance policy is monotone if $\ProbAcc$ is (weakly) increasing, i.e.~$\ProbAcc(r')\ge\ProbAcc(r)$ for any $r'\ge r$.
\end{definition}



\begin{definition}
    We say a review policy is monotone if $\ProbRev(\phi'_i)\ge \ProbRev(\phi_i)$ for any $\phi'_i, \phi_i\in \Phi_i$ such that $\phi'_i\succeq\phi_i$ in any round $i$.
\end{definition}

A monotone acceptance policy rewards a paper with a higher review score by accepting it with a higher probability; a monotone review policy rewards a higher-ordered review state with an increased probability of reviewing the paper in the next round. However, the requirements for the state transition mapping are a little more complicated. In particular, we utilize the concept of stochastic dominance. 

\begin{definition}
Let $\bm{X}$ and $\bm{Y}$ be two $m$-dimension random vectors of review states $\bm{X}, \bm{Y}\in\Phi^m$ for some review state space $\Phi$. We say $\bm{X}$ first-order stochastic dominates $\bm{Y}$ if $\Pr(\bm{X}\succeq \bm{\phi})\ge \Pr(\bm{Y}\succeq \bm{\phi})$ for any $\bm{\phi}\in\Phi^m$.    
\end{definition}

For simplicity, let $\tilde{\mu}(r_1, r_2|\phi_i) = \mu_{i+1}(\mu_{i}(\phi_i, r_1), r_2)$ be the state distribution in round $i+1$ conditioned on having review state $\phi_i$ in round $i-1$, and having review scores $r_1$ and $r_2$ in round $i$ and round $i+1$ respectively.

\begin{definition}\label{def:mu_monotone}
    We say the state transition mapping $\bm{\mu}$ is \emph{monotone} if for any review round $i\in [n-1]$,
    \begin{enumerate}
        \item it is monotone in score: For any state $\phi_i\in\Phi_i$, $\mu_i(\phi_i, r')$ first-order stochastic dominates $\mu_i(\phi_i, r)$ for any $r'\ge r$;
        \item it is monotone in state: For any review score $r$, $\mu_i(\phi', r)$ first-order stochastic dominates $\mu_i(\phi, r)$ for any state $\phi' \succeq \phi$;
        \item it is monotone in ordering: For any state $\phi_i\in\Phi_i$ and review scores $r'\ge r$, $\tilde{\mu}(r', r|\phi_i)$ first-order stochastic dominates $\tilde{\mu}(r, r'|\phi_i)$. Furthermore, for any $r_1\le r_2\le r_3\le r_4$ such that $r_1+r_4 = r_2+r_3$, let $X\sim \tilde{\mu}(r_4, r_1|\phi_i)$, $Y\sim \tilde{\mu}(r_2, r_3|\phi_i)$, $X'\sim\tilde{\mu}(r_1, r_4|\phi_i)$ and $Y'\sim\tilde{\mu}(r_3, r_2|\phi_i)$. Let $\bar{Z} = \max(X, Y)$, $\underbar{Z} = \min(X, Y)$, $\bar{Z}' = \max(X', Y')$, $\underbar{Z}' = \min(X', Y')$.\footnote{Here, the $\max$ and $\min$ function select the higher-ordered state and the lower-ordered state respectively.} Then, $(\bar{Z}, \underbar{Z})$ first-order stochastic dominates $(\bar{Z}', \underbar{Z}')$.
    \end{enumerate} 
\end{definition}



The monotonicities of $\bm{\mu}$ in score and state are intuitive which suggest that the state transition mapping results in a better state when the review score is higher and the review state is higher-ordered, respectively. The monotonicity in ordering deals with the cases where the review scores in two rounds are swapped. First, it requires the review state distribution to be better if the higher review score is put earlier. Furthermore, when there are four ordered review scores and a pair of review states, the monotonicity in ordering suggests that putting the largest score $r_4$ earlier in round $i$ with the lowest score $r_1$ in round $i+1$ and putting $r_2$ in round  $i$ with $r_3$ in round $i+1$ leads to better state distributions than swapping the review scores in round $i$ and $i+1$ in this ordering. 
We are ready to present the main theorem.
\begin{theorem}\label{thm:sufficient_cond}
    The sequential review mechanism $\mathcal{M}^s=(\ProbAcc, \ProbRev, \bm{\mu})$ is truthful if $\ProbAcc$, $\ProbRev$ and $\bm{\mu}$ are monotone.
\end{theorem}

At a high level, the proof follows by coupling the realizations of the review noise. Then, due to the monotonicity of all three components of the sequential review mechanism, flipping the true order of any two papers will result in a review state that is always dominated by truthful reporting.

\begin{proof}
We first show that the proof can be reduced to the setting in which the author possesses perfect information regarding the true quality of her papers. Intuitively, this is because the author's noise affects her reasoning about the ranking of papers in the same way as the review noise. The argument is as follows. Suppose the author observes the true quality of her papers, and her best response is to truthfully rank her papers based on $q_i$ for any i.i.d.~review noise $\epsilon_i$. Then, we can always construct a new review noise $\tilde{\epsilon}_i = \epsilon_i - \xi_i$ such that the review score of paper $i$ is $r_i = s_i + \hat{\epsilon}_i$. By assumption, $\hat{\epsilon}_i$ is i.i.d.~for all papers. Consequently, it straightforwardly follows that the author's best response is to rank her papers based on the true order of $s_i$ in this case.
Therefore, in the following proof, it is sufficient to focus on the setting where the author knows the true qualities.

Consider an arbitrary untruthful ranking $\pi$. Without loss of generality, suppose under $\pi$, there exists a pair of papers such that the lower-quality one is ranked higher than the higher-quality one. Otherwise, $\pi$ only permutes papers with the same quality and is trivially (weakly) dominated by the truthful ranking. Suppose paper $j$ and $k$ are ranked in adjacent places under $\pi$, however $\pi(j) > \pi(k)$ but  $q_j > q_k$---that is, paper $j$ is better than paper $k$ but ranked slightly worse.  Let $\pi(k) = t$, so $\pi(j) = t + 1$. The existence of such two papers is guaranteed by the design of $\pi$. Then, consider a ranking $\pi'$ which switches the ordering of paper $j$ and $k$ such that $\pi'(j) = t$ and $\pi'(k) = t + 1$.   We will show that the author always weakly prefers to report $\pi'$ instead of $\pi$. Note that this is sufficient to guarantee truthfulness because this shows the author always prefers to rank a higher-quality paper ahead of an adjacent lower-quality paper.  

We apply the idea of coupling to show $\E[U_a(\pi')]\ge \E[U_a(\pi)]$, where the expectation is taken over the randomness of the review noise and the sequential review mechanism. 

To accomplish this, we will present a one-to-one matching where in any realization of outcomes,  every paper accepted in $\pi$ can be matched with a paper accepted in $\pi'$ with the same or higher quality.  

We now construct such a matching by coupling the randomness of the review noise, the state transition mapping, the review policy, and the acceptance policy.  

First, we couple the realizations of the review noise under both permutations. In particular, consider two noise vectors $\bm{\epsilon}$ and $\hat{\bm{\epsilon}}$ which are identical everywhere except that the noises in round $t$ and $t+1$ are swapped. Let $\bm{\epsilon}_t = \hat{\bm{\epsilon}}_{t+1} = \epsilon_t$ and $\bm{\epsilon}_{t+1} = \hat{\bm{\epsilon}}_t = \epsilon_{t+1}$. We couple the review noise under the two permutations such that whenever the review noise is $\bm{\epsilon}$ under $\pi$, the review noise is $\hat{\bm{\epsilon}}$ under $\pi'$; while whenever the review noise is $\hat{\bm{\epsilon}}$ under $\pi$, the review noise is $\bm{\epsilon}$ under $\pi'$.  This is possible because all the review noise terms are i.i.d.. Then, we will show that conditioned on the review noise, $\E[U_a(\pi')|\bm{\epsilon}] + \E[U_a(\pi')|\hat{\bm{\epsilon}}] \ge \E[U_a(\pi)|\bm{\epsilon}] + \E[U_a(\pi)|\hat{\bm{\epsilon}}]$, where the expectation is taken over the randomness of the sequential review mechanism.

Fixing the review noise fixes the review score. Let $\bm{r}$ and $\hat{\bm{r}}$ be the review score vectors under $\pi$ with review noise $\bm{\epsilon}$ and $\hat{\bm{\epsilon}}$ respectively, and let $\bm{r}'$ and $\hat{\bm{r}}'$ be the corresponding terms under $\pi'$. Therefore, $\bm{r}$, $\hat{\bm{r}}$, $\bm{r}'$ and $\hat{\bm{r}}'$ are identical except on the following entries.

\begin{figure*}[htb]
     \centering
     \begin{subfigure}[b]{0.7\textwidth}
         \centering
         \includegraphics[width=\textwidth]{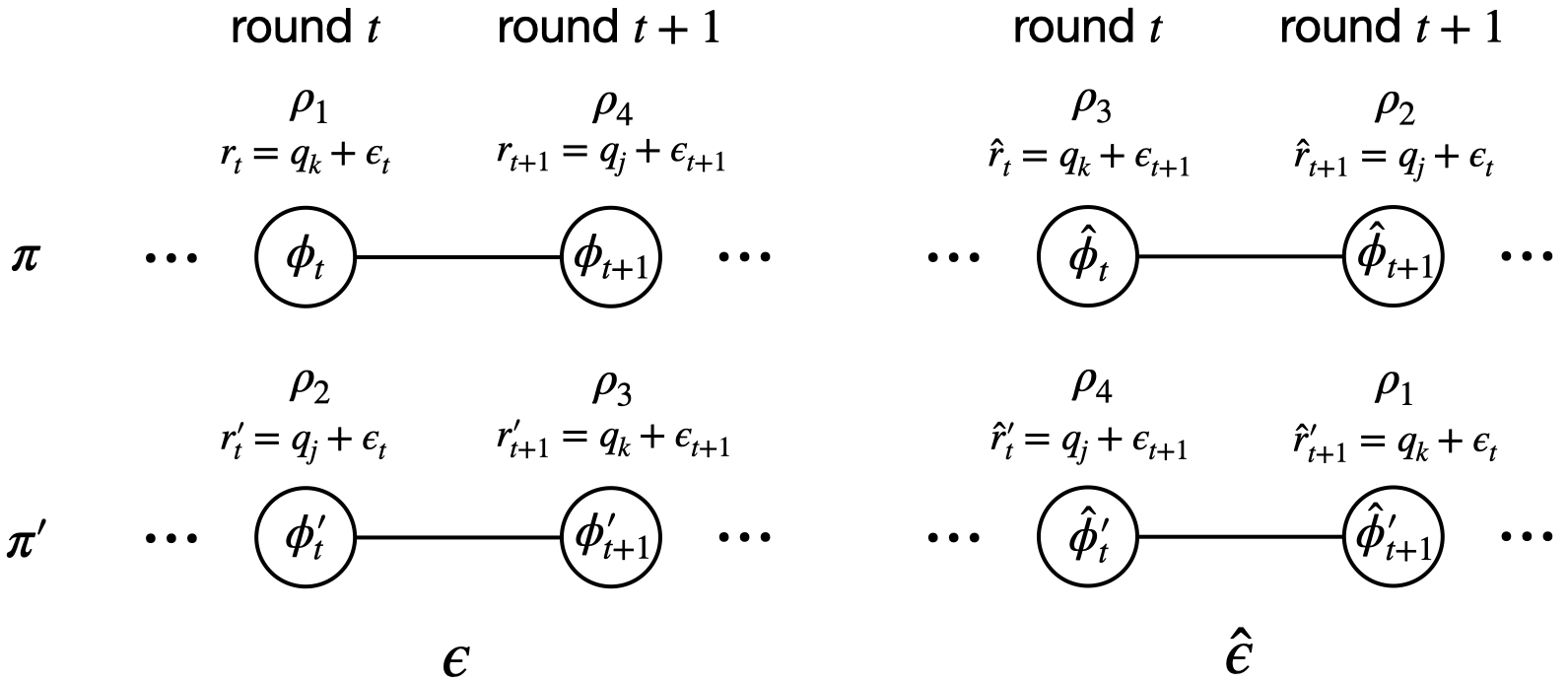}
     \end{subfigure}
     \hfill
\end{figure*}

We want to show that the sum of the expected utility of having review scores $\bm{r}'$ and $\hat{\bm{r}}'$ (review scores under $\pi'$) is larger than that of having $\bm{r}$ and $\hat{\bm{r}}$ (review scores under $\pi$). We prove this by discussing two cases each with its own coupling.   

\paragraph{Case 1.} The first case corresponds to the review noise vectors such that $q_j -q_k \geq |\epsilon_t -\epsilon_{t+1}| $. 
Here, permutation $\pi'$ has a larger review score that $\pi$ in round $t$ for both noise vectors.   We show that $\E[U_a|\bm{r}'] \ge \E[U_a|\hat{\bm{r}}]$ and $\E[U_a|\hat{\bm{r}}'] \ge \E[U_a|\bm{r}]$. 
This is sufficient to prove case 1 because permutation $\pi'$ always leads to a higher expected utility under every realization of review noise.

We present the proof of $\E[U_a|\bm{r}'] \ge \E[U_a|\hat{\bm{r}}]$ while the proof of $\E[U_a|\hat{\bm{r}}'] \ge \E[U_a|\bm{r}]$ is analogous. Let $\bm{\phi}'$ and $\bm{\hat{\phi}}$ be the vectors of review states under $\bm{r}'$ (corresponding to $\pi'$) and $\hat{\bm{r}}$ (corresponding to $\pi$) respectively, where each entry is the review state in round $i$. Let $\rho_1 = r_t=\hat{r}'_{t+1}$ and $\rho_4 = \hat{r}'_t=r_{t+1}$. We first show that in any round $i$, the distribution of $\phi'_i$ first-order stochastic dominates that of $\hat{\phi}_i$. 


    The argument follows by discussing the following four cases which depend on where $i$ is.
    \begin{enumerate}
    \item If $i\le t$, the statement holds because the review scores are identical. Then, the review state distribution must be identical for any round before $t$.
    \item If $i=t + 1$, the statement holds because $\bm{\mu}$ is monotone in score. Since $r'_t \ge r_t$, permutation $\pi'$ has a larger score than $\pi$ in round $i$.
    \item If $i=t + 2$, the statement holds because $\bm{\mu}$ is monotone in ordering. In particular, because conditioned on the same review state $\phi_t$, $\tilde{\mu}(r', r|\phi_t)$ first-order stochastic dominates $\tilde{\mu}(r, r'|\phi_t)$ for every $r'\ge r$. 
    \item If $i>t + 2$, the statement holds because $\bm{\mu}$ is monotone in state, so this holds by induction.  
\end{enumerate}

Now, we couple the review states such that $\phi'_i\succeq \hat{\phi}_i$ for any round $i$. Then, we further couple the randomness of $\ProbRev$. Let $X_i', \hat{X}_i\in [0,1]$ be random variables of two independent samples from $[0,1]$ uniformly at random such that the review process terminates in round $i$ if and only if $X_i'> \ProbRev(\phi'_i)$ and $\hat{X}_i> \ProbRev(\hat{\phi}_i)$ under $\pi'$ and $\pi$ respectively. We couple these random variables such that $X'_i = \hat{X}_i$ for any round $i$. Because $\ProbRev$ is monotone and $\phi'_i\succeq \hat{\phi}_i$ for any $i$, $\ProbRev(\phi'_i)\ge \ProbRev(\hat{\phi}_i)$. Therefore, in any realization of outcomes, whenever the review process terminates under $\pi'$, it also terminates under $\pi$. 

Fixing an outcome for which papers are reviewed, let $\hat{I}$ and $I'$ be the set of rounds in which papers are reviewed under $\pi$ with $\bm{\hat{\mu}}$ and $\pi'$ with $\bm{\mu}$ respectively. 
We would like to show there exists a one-to-one matching $\sigma: \hat{I} \rightarrow I'$ such that for all $i \in \hat{I}$, the paper reviewed in round $\sigma(i) \in I'$ under $\pi'$ has a weakly higher quality and review score than the paper reviewed in round $i$ under $\pi$.

We break up the proof into two cases depending on if $t + 1 \in I'$ (i.e.~the review process under $\pi'$ and $\bm{\epsilon}$ terminates before round $t+1$). If $t + 1 \not\in I'$ (e.g. $\pi$ terminates before round $t+1$), then $\sigma$ is just the identity map.  We know that the process in $\pi$ also terminates before round $t+1$, and that the state, paper quality, and review score in $\pi'$ coordinate-wise dominate those of $\pi$.  

If the process of $\pi'$ terminates after round $t+1$, we let $\sigma(t) = t+1$, $\sigma(t+1) = t$, and $\sigma(i) = i$ for all other $i \in \hat{I}$.  Again, for all $i \in \hat{I}$, the paper in round $\sigma(i)$ under $\pi'$ has both weakly higher quality and weakly higher scores than paper in round $i$ in $\pi$.  We can first verify this for $i \in \{t, t+1\}$, where the papers are swapped under $\pi$ and $\pi'$. If the papers in these two rounds are reviewed in $\pi$, they are also reviewed in $\pi'$, and, by inspection, have the same quality and review score. For $i \not\in \{t, t+1\}$, the claim holds because $\bm{r}'$ and $\hat{\bm{r}}$ are identical, the corresponding qualities of papers are equal, and any paper reviewed in $\pi$ is also reviewed in $\pi'$.





Finally, we couple the randomness of $\ProbAcc$. Let $Z_i', \hat{Z}_i\in [0,1]$ be random variables of two independent samples from $[0,1]$ uniformly at random such that the paper in round $i$ is accepted if and only if $Z'_i\le \ProbAcc(r'_i)$ and $\hat{Z}_i\le \ProbAcc(\hat{r}_i)$ under $\pi'$ and $\pi$ respectively. Fix a realization of the review outcome, let $\sigma$ be the one-to-one matching such that whenever the paper in round $i$ is reviewed under $\pi$, the paper in round $\sigma(i)$ is reviewed under $\pi'$. We have shown that such a $\sigma$ always exists. Then, we couple the random variables such that $Z'_{\sigma(i)} = \hat{Z}_i$ for any round $i$. Because $\ProbAcc$ is monotone, whenever a paper is reviewed and accepted under $\pi$, a paper with a higher review score and a higher quality is reviewed and accepted under $\pi'$. This completes the proof of case 1.

\paragraph{Case 2.}
The second case corresponds to the review noise vectors such that $q_j -q_k < |\epsilon_t -\epsilon_{t+1}| $. In this case, we cannot isolate a pair of expectations as we have done in case 1 but have to prove the theorem by directly showing that $\E[U_a|\bm{r}'] + \E[U_a|\hat{\bm{r}}'] \ge \E[U_a|\bm{r}] + \E[U_a|\hat{\bm{r}}]$.

We couple the randomness of $\bm{\mu}$ and $\ProbRev$. For any realization of the review outcomes, we compare the set of the papers that are reviewed under $\pi$ (for both $\bm{\epsilon}$ and $\hat{\bm{\epsilon}}$ so there are at most $2n$ reviewed papers) and the set of papers that are reviewed under $\pi'$. For the set of the reviewed papers, we want to find a one-to-one matching such that every paper that is reviewed under $\pi$ is matched with a paper with both a higher quality and a higher score that is reviewed under $\pi'$.

We first introduce some notations. Let $\bm{\phi}'$, $\hat{\bm{\phi}}'$, $\bm{\phi}$ and $\hat{\bm{\phi}}$ be the vector of review states conditioned on $\bm{r}'$, $\hat{\bm{r}}'$, $\bm{r}$ and $\hat{\bm{r}}$ respectively. Let $\rho_1 = r_t=\hat{r}'_{t+1}$, $\rho_2 = r'_t=\hat{r}_{t+1}$, $\rho_3 = \hat{r}_t=r'_{t+1}$ and $\rho_4 = \hat{r}'_t=r_{t+1}$. We have $\rho_1\le \rho_2<\rho_3\le \rho_4$. Furthermore, let $X_i, \hat{X}_i\in [0,1]$ be random variables of two independent samples from $[0,1]$ uniformly at random such that the review process terminates in round $i$ if and only if $X_i> \ProbRev(\phi_i)$ and $\hat{X}_i > \ProbRev(\hat{\phi}_i)$ under permutation $\pi$ respectively. Let $X'_i$ and $\hat{X}'_i$ be the analogous variables under $\pi'$. We couple these random variables such that $X_i = X'_i = \hat{X}_i = \hat{X}'_i$ for any round $i$. Now, we couple the review states and construct the matching of papers which will allow us to couple the acceptances just as before.  
\begin{enumerate}
    \item For any $i\le t-1$, we couple the review states such that $\phi_i \sim \phi'_i \sim \hat{\phi}_i \sim \hat{\phi}'_i$. In this case, because the papers and review scores in the first $t-1$ rounds are identical for any discussed review process, all the discussed review processes must terminate in the same round (if they ever terminate before round $t$), resulting in the same set of papers being reviewed under the two permutations. Therefore, we construct the matching such that every paper in round $i$ under $\pi$ is paired with the same paper in round $i$ under $\pi'$.

    \item For papers in round $t$ and $t+1$, we construct the matching by discussing different cases. First note that we can couple the review states such that $\phi_t \sim \phi'_t \sim \hat{\phi}_t \sim \hat{\phi}'_t$. Therefore, we only have to discuss the papers in round $t$ and round $t+1$ when $\phi_t \sim \phi'_t \sim \hat{\phi}_t \sim \hat{\phi}'_t\succ \omega$. Otherwise, because $X_t=X'_t=\hat{X}_t=\hat{X}'_t$, none of the papers in round $t$ and $t+1$ are reviewed and the matching in the previous step guarantees the same set of papers are reviewed under two permutations.
    
    Now, we couple the review states. Because $\rho_1 < \rho_2$ and $\rho_3 < \rho_4$, we can couple the realizations of the states in round $t+1$ so that $\phi'_{t+1}\succeq \phi_{t+1}$ and $\hat{\phi}'_{t+1}\succeq \hat{\phi}_{t+1}$ since $\bm{\mu}$ is monotone in score. We construct the matching as follows.
    \begin{itemize}
             \item  If $\phi'_{t+1} \sim \hat{\phi}'_{t+1} \sim \omega$, only the papers in the first $t$ rounds are reviewed. We match the paper with score $r_t = \rho_1$ with the paper and score $r'_t = \rho_2$ and match the paper with score $\hat{r}_t = \rho_3$ with the paper with score $\hat{r}'_t = \rho_4$. 
             \item If $\phi'_{t+1} \succ \hat{\phi}'_{t+1} \sim \omega$, papers with review scores $\hat{r}'_{t+1}$ and $\hat{r}_{t+1}$ are not reviewed, the paper with score $r'_{t+1}$ is reviewed, and the paper with score $r_{t+1}$ may or may not be reviewed. In this case, we apply the following matching: the paper with score $r_t = \rho_1$ is matched with the paper with score $r'_t = \rho_2$, the paper with score $r_{t+1} = \rho_4$ is matched with the paper with score $\hat{r}'_t = \rho_4$, and the paper with score $\hat{r}_t = \rho_3$ is matched with the paper with score $r'_{t+1} = \rho_3$.
             \item If $\hat{\phi}'_{t+1} \succ \phi'_{t+1} \sim \omega$, papers with review scores $r'_{t+1}$ and $r_{t+1}$ are not reviewed, the paper with score $\hat{r}'_{t+1}$ is reviewed while the paper with score $\hat{r}_{t+1}$ may or may not be reviewed. In this case, we apply the following matching: the paper with score $r_t = \rho_1$ is matched with the paper with score $\hat{r}'_{t+1} = \rho_1$, the paper with score $\hat{r}_t = \rho_3$ is matched with the paper with score $\hat{r}'_t = \rho_4$, and the paper with score $\hat{r}_{t+1} = \rho_2$ is matched with the paper with score $r'_t = \rho_2$.
             \item If $\hat{\phi}'_{t+1} \succ \omega$ and $\phi'_{t+1} \succ \omega$, all papers in the first $t+1$ rounds are reviewed under $\pi'$, while the papers in round $t+1$ may or may not be reviewed under $\pi$. In this case, we match the paper with review score $\rho_j$ under $\pi$ with the paper with review score $\rho_j$ under $\pi'$ for $j\in \{1,2,3,4\}$.
    \end{itemize}
    It is easy to see that in any case, the above matching guarantees that $\pi'$ reviews papers with pairwise higher qualities and scores than $\pi$.
    \item For any $i\ge t+2$, we couple the review states after round $t+2$. Let $\tilde{\phi}_i = (\max(\phi_i, \hat{\phi}_i), \min(\phi_i, \hat{\phi}_i))$ be a pair of review states in round $i$ under permutation $\pi$ such that the first entry is always the better state between $\phi_i$ and $\hat{\phi}_i$ while the second entry is the worse one. Let $\tilde{\phi}'_i$ be the analogous notation for $\pi'$. 
    We will show that the distribution of $\tilde{\phi}'_i$ first order stochastic dominates the distribution of $\tilde{\phi}_i$.
    \begin{itemize}
        \item If $i=t + 2$, the statement holds because $\bm{\mu}$ is monotone in ordering. Suppose without loss of generality that $\epsilon_t \le \epsilon_{t+1}$ while the proof of the other case is analogous. Then, by definition, $\phi'_{t+2}\sim \tilde{\mu}(\rho_2, \rho_3)$, $\hat{\phi}'_{t+2}\sim \tilde{\mu}(\rho_4, \rho_1)$, $\phi_{t+2}\sim \tilde{\mu}(\rho_1, \rho_4)$, $\hat{\phi}_{t+2}\sim \tilde{\mu}(\rho_3, \rho_2)$. Therefore, conditioned on the same review state $\phi_{t}$, the distribution of $\tilde{\phi}'_{t+2}$ first-order stochastic dominates that of $\tilde{\phi}_{t+2}$. This is why we need the definition of the monotonicity of ordering (\cref{def:mu_monotone}).
        \item If $i>t + 2$, the statement holds because $\bm{\mu}$ is monotone in state. This analysis is analogous to the analysis in case 1.
    \end{itemize}
    Now, we can couple the review states such that $\tilde{\phi}'_i\succeq \tilde{\phi}_i$ for any $i\ge t+2$. We construct the one-to-one matching such that every paper in round $i$ under $\pi$ is matched with the same paper in round $i$ under $\pi'$. Because we have coupled $X$s and $\ProbRev$ is monotone, for every paper in round $i\ge t+2$ that is reviewed under $\pi$, our matching guarantees that the same paper will be reviewed under $\pi'$.
\end{enumerate}

Lastly, we complete the proof by coupling the randomness of the acceptance policy $\ProbAcc$. This straightforwardly follows from the same argument in case 1. Therefore, every paper that is accepted under $\pi$ can be matched with an accepted paper with weakly higher quality and score under $\pi'$.

\end{proof}

\subsection{The memoryless Coin-Flip Mechanism}
\label{subsec:coin-flip}

Here, we provide an example of how to use our framework to design a truthful sequential review mechanism. In this example, the acceptance of a paper in round $i$ will guarantee the review of the paper in the next round; while if a paper is rejected, the mechanism will review the paper in round $i+1$ with probability $\rho(r_i)$ determined by the review score of the paper in round $i$. Furthermore, the mechanism always reviews the first paper. We call this mechanism the \emph{memoryless coin-flip mechanism}. 
 
To map this mechanism into the sequential review mechanism framework, we first define the review states. Consider $\phi_i = (\alpha, \gamma) \in\{0,1\}^2$ for any round $i>1$, where the first entry indicates whether the paper in round $i-1$ has been accepted (with $1$ being accepted and $0$ being rejected), and the second entry indicates whether a (biased) coin flip has an outcome of heads (with $1$ being heads and $0$ being tails). If at least one of $\alpha$ and $\gamma$ is $1$, the next paper will be reviewed; otherwise, the review process terminates. That is, $(0,0) = \omega$. Note that the distribution of the review state in round $i$ only depends on the review score of round $i-1$, which implies that the review policy is ``memoryless''. Furthermore, because the first paper is always reviewed, let $\phi_1 = (1,\gamma)$. In this example, the ordering of the review states is $(1,1)\sim (1,0)\sim (0,1) \succ \omega$.

Now, fixing an acceptance policy $\ProbAcc$, we design the review policy $\ProbRev^{cf}$ and the state transition mapping $\mu_i^{cf}$ for each round $i$ as follows.

\begin{equation*}
    \ProbRev^{cf}((\alpha, \gamma)) =
  \begin{cases}
  1&\mbox{if $(\alpha, \gamma) \neq \omega$},\\
  0&\mbox{otherwise}.\\
  \end{cases}
\end{equation*}

\begin{equation*}
    \mu^{cf}_i((\alpha_i,\gamma_i), r_{i+1}) =
  \begin{cases}
  (1, 1)&\mbox{ with probability $\ProbAcc(r_{i+1})\cdot \rho(r_{i+1})$},\\
  (1, 0)&\mbox{ with probability $\ProbAcc(r_{i+1})\cdot (1-\rho(r_{i+1}))$},\\
  (0, 1)&\mbox{ with probability $(1-\ProbAcc(r_{i+1}))\cdot \rho(r_{i+1})$},\\
  \quad\omega&\mbox{ with probability $(1-\ProbAcc(r_{i+1}))\cdot (1-\rho(r_{i+1}))$},
  \end{cases}
  \end{equation*}
where $\rho: \mathbb{R}\rightarrow [0,1]$ maps from a review score to a probability of reviewing the next paper. Note that the memoryless coin-flip mechanism reduces to the parallel review mechanism (which unconditionally reviews all papers) when $\rho(r) = 1$ for any review score $r$, and reduces to the naive sequential review mechanism when $\rho(r) = 0$ for any $r$.
Because the transition mapping does not depend on the index of round, we thus omit the subscript of $\mu$ while discussing the memoryless coin-flip mechanism.

Now, we show that the memoryless coin-flip mechanism is truthful by mapping it to the sufficient conditions of truthfulness as shown in \cref{thm:sufficient_cond}. We defer the proof to \cref{app:CF_proof}.

\begin{prop}\label{thm:cf-truth}
    The memoryless coin-flip mechanism $(\ProbAcc, \ProbRev^{cf}, \bm{\mu}^{cf})$ with function $\rho$ is truthful if $\ProbAcc$ is monotone and $\rho$ is non-decreasing.
\end{prop}

\subsection{The Credit Pool Mechanism}

Another example of the sequential reviewing framework implements the idea of a reputation system. 
Suppose the conference keeps a record of a credit pool, which is initialized at $B_1\ge 0$. For every reviewed paper, the mechanism will increase (or decrease) the credit pool by a credit score which is determined by the review of that paper. Let $\beta: \mathbb{R}\rightarrow \mathbb{R}$ be a credit function that maps from a review score to a review credit. Note that $\beta$ can be negative, indicating a punishment of papers with low review scores. The credit pool mechanism reviews the paper in round $i$ if and only if $B_i\ge 0$.  Details of the credit pool mechanism are presented in \ref{alg:credit_pool}. 


\begin{mechanism}[htbp]
 \KwIn{The ordered review scores $\bm{r}$ with $r_t$ being the score of the paper ranked by the author in round $t$, the acceptance policy $\ProbAcc$, the initial credit $B_1$ and the credit function $\beta$.}
 \KwOut{The acceptance decision of each of the $n$ papers $\bm{\alpha}$.}
 Initialize the credit pool as $B=B_1\ge 0$ and $t=1$;\\
 \While{$B\ge 0$}{
    Flip a biased coin with a probability of heads equal to $\ProbAcc(r_t)$, and the outcome is $\gamma_t$ (with 1 being heads and 0 being tails);\\
    \eIf{$\gamma_t = 1$}{
    $\alpha_t = 1$;\\
    }{
    $\alpha_t = 0$;\\
    }
    $B \mathrel{+} = \beta(r_t)$;\\
    $t \mathrel{+} = 1$;\\
 }
 $\alpha_j = 0$ \textbf{for} $j\in \{t, \ldots,n\}$.

 \caption{The credit pool mechanism.}\label{alg:credit_pool}
\end{mechanism}

Speaking in the language of the sequential review mechanism framework, we first define the review state as the credit in the pool, i.e.~$\phi_i = B_i$ for round $i$. States are ordered in terms of their credit, however, and any negative credit will correspond to the termination state $\omega$. More formally, the ordering of the states is $B'\succeq B \succ \omega$ for any $B'\ge B\ge 0$. We define the review policy as
\begin{equation*}
    \ProbRev^{cp}(B) =
  \begin{cases}
  1&\mbox{if $B\ge 0$},\\
  0&\mbox{otherwise}.\\
  \end{cases}
\end{equation*}
Furthermore, the state transition mapping $\bm{\mu}^{cp}$ is deterministic such that $\bm{\mu}^{cp}(\omega, r) = \omega$ with probability of 1 and conditioned on $B_i\ge 0$,
\begin{equation*}
    \mu^{cp}(B_i, r_{i}) = B_i + \beta(r_{i}) \quad\text{ with probability 1}.
\end{equation*}
Again, we can omit the subscript $i$ since $\mu^{cp}$ does not depend on the review round. The following theorem indicates sufficient conditions for a credit pool mechanism to be truthful.

\begin{prop}\label{thm:cp_truth}
    The credit pool mechanism $(\ProbAcc, \ProbRev^{cp}, \bm{\mu}^{cp})$ with credit function $\beta$ is truthful if $\ProbAcc$ is monotone and $\beta$ is increasing and convex.
\end{prop}

We defer the proof to \cref{app:CP_proof}.
Intuitively, the credit pool mechanism is truthful because a higher-quality paper has both a larger acceptance probability and a larger chance of increasing the review credit. Therefore, any untruthful permutation will, more likely, result in an earlier termination of the review process which always leads to a smaller utility in expectation.

\subsection{Non-truthful Examples}
We further provide a discussion on several heuristic mechanisms that are not truthful, which illustrates what can go wrong.  

\paragraph{Ranking bundles of papers.}
A plausible generalization of the sequential review mechanism is to require the author to divide her papers into several bundles and provide an order of bundles. For example, an author with four papers can be asked to report two papers as the (equally) best papers and two as the (equally) bottom papers. Then, the mechanism will review the papers in the next bundle if and only if certain conditions are satisfied based on the review scores of the papers in the previous bundle (e.g., the acceptance of at least one paper). 

However, such an idea is not truthful in general. Consider a simple example where the mechanism reviews two papers at a time and will review the next two papers if and only if at least one of the two papers is accepted. Suppose an author has six papers to submit, and two of them are outstanding and basically guaranteed acceptance while the remaining four papers are of borderline quality. Suppose the author knows the true qualities. Now, the author has the incentive to put one of the two outstanding papers in the third place and rank one borderline paper in the second place. This strategy almost certainly ensures that all six papers can be reviewed. However, if the author reports the true ranking, it is more likely that only the top four papers will be reviewed, while the bottom two papers will be rejected without review. 

\paragraph{The limited credit pool mechanism.}
We present a non-truthful variant of the credit pool mechanism with a credit pool of a limited size.  The growth of the pool is halted once the accumulated credit reaches a predefined threshold. When the threshold is reached, even if subsequent papers receive high review scores, they will no longer contribute more credit to the pool.

The limited credit pool mechanism is not truthful, because authors are better-off to have their high-quality papers reviewed only when the credit pool is not full. For example, suppose an author has a lot of borderline papers with negative quality and two outstanding papers that are basically guaranteed acceptance and any one of them can fill the credit pool. Now, compared with truth-telling, the author is better off ranking one outstanding paper in the first place followed by one or two borderline papers and then the other outstanding paper. This untruthful ranking ensures that the author secures more reviews for her borderline papers than the truthful ranking.

\section{Evaluating the Threshold Sequential Review Mechanism}
\label{sec:opt_mechanism}


The previous section identifies a space of truthful mechanisms, while this section focuses on empirically evaluating the performance of the threshold sequential review mechanism. We evaluate a mechanism from two dimensions: conference utility and review burden. The former is measured by the sum of accepted papers' quality, and the latter quantifies the number of reviewed papers, while both are normalized by the total number of submitted papers. For both computational and practical considerations, we focus on the threshold sequential review mechanism, a special case of the memoryless coin-flip mechanism discussed in \cref{subsec:coin-flip}. In comparison, we use the threshold parallel review mechanism as a baseline and the threshold isotonic mechanism with the underlying ranking information as an unreachable upper bound. Therefore, when we refer to a mechanism in this section, unless otherwise specified, we simply mean the threshold implementation of it.

We conduct experiments on both a simple model with Gaussian review noise and a more complicated real-data estimated model where each paper has multiple integer-valued review scores.\footnote{The implementation of our experiments is available at \url{https://github.com/yichiz97/Sequential-Review}.}
Through Monte-Carlo simulations, we show that the sequential review mechanism can achieve conference utility that is competitive compared to the upper bound and can significantly reduce the review burden. Moreover, in the real-data estimated model, we show that the sequential review mechanism can save more than 20\% review burden compared with the baseline conditioned on a weakly better conference utility.
Because of a lack of ground truth data, in this section, we simply assume that authors observe the perfect information about the true quality of their papers.

\subsection{Mechanisms of Comparison}

Here, we introduce how we implement and optimize the sequential review mechanism, the parallel review mechanism, and the isotonic mechanism in our experiments. For all three types of mechanisms, we focus on the threshold acceptance policy, i.e.~a paper is accepted if and only if its (modified) review score is no less than a threshold.  

\paragraph{The threshold sequential review mechanism.} 
We restrict our attention to a special case of the sequential review mechanism with two special parameters $\tau^s_{acc}$ and $\tau^s_{rev}$ where  $\tau^s_{acc} \geq \tau^s_{rev}$.  The threshold sequential review mechanism is a memoryless coin-flip mechanism with $\ProbAcc$ and $\rho$ being threshold functions: $\ProbAcc(r) = 1$ if $r\ge \tau^s_{acc}$ and 0 otherwise; $\rho(r)=1$ if $r\ge \tau^s_{rev}$ and 0 otherwise. In words, the threshold sequential review mechanism accepts a paper, conditioned on it being reviewed, if its review score is larger than a threshold $\tau^s_{acc}$; and it reviews the next paper if the review score is larger than a lower threshold $\tau^s_{rev}$. We study the threshold sequential review mechanism mainly because of its simplicity, making it easy to optimize and implement in practice.  This illustrates the power of the sequential review mechanism even in a simplified form.

\paragraph{The threshold parallel review mechanism.}
We consider the parallel review mechanism as a baseline for comparison.
Under a parallel review mechanism, all papers are guaranteed to be independently reviewed. The parallel review mechanism operates without soliciting the ranking information from authors. The threshold parallel review mechanism is characterized by the acceptance threshold $\tau^\ProbAcc$ where papers are accepted if and only if their review scores are no less than $\tau^\ProbAcc$. Note that the threshold parallel review mechanism is a special case of the threshold sequential review mechanism by setting $\tau^s_{rev} = -\infty$.

\paragraph{The threshold isotonic mechanism with true ranking information.}
In general, the isotonic mechanism cannot truthfully elicit the ranking information from authors in our setting where the author's utility is not convex with respect to the review score. However, we assume the author still reports the ranking of their papers' quality truthfully and uses the isotonic mechanism with this information as the upper bound. In particular, given the original review scores, the isotonic mechanism modifies them by solving the isotonic regression conditioned on the author's ranking (details are presented in \cref{app:isotonic}). Then, a threshold acceptance policy is applied to the modified review scores such that the paper with score $r$ is accepted if and only if $r\ge \tau^i_{acc}$.

To emphasize the importance of truthfulness, we note that in the absence of truthful ranking information, the isotonic mechanism can yield a conference utility that is even worse than the baseline. We illustrate this with an example presented in \cref{app:isotonic}.  However, assuming truthfulness gives us a (likely unachievable) benchmark.   


\subsection{Gaussian Review Noise}

We first introduce a simple yet intuitive model where each paper is associated with a single review score, which is the true quality plus an additive Gaussian review noise. The simplicity of this model offers computational convenience, enabling us to efficiently explore the influence of various model parameters on the performance of the sequential review mechanism. Note that under this simple model, we again assume that there is only one author so as to present our results in a clear manner.

\subsubsection{Model and Experiment Setup}
The \emph{Gaussian review setting} $\varphi^G = (n, \mu_q, \sigma_q, \sigma_r)$ is defined by four parameters.
The conference has a Gaussian prior of the quality of the author's papers with parameters $\mu_q$ and $\sigma_q$. Thus, the author draws $n$ i.i.d.~samples from the Gaussian distribution $\mathcal{N}(\mu_q, \sigma_q)$ as the qualities of her papers. Then, the author sends the conference a ranked list of her papers in terms of their quality (hence the conference observes the true ranking of the quality of the papers). Let $\bm{q}$ be the ordered vector of paper qualities such that $q_i\ge q_j$ for any $1\le i\le j\le n$. Next, the conference draws an i.i.d.~review noise term from a zero-mean Gaussian distribution $\epsilon_i\sim \mathcal{N}(0, \sigma_r)$ for each $i\in [n]$. Finally, the conference observes one review score for each paper, $r_i = q_i+\epsilon_i$. 

Given a setting $\varphi^G$ and a mechanism $\mathcal{M}$, we use the Monte-Carlo method to estimate the expected conference utility. In particular, we draw \sisetup{group-separator={,}}\num{10000} samples of $\bm{q}$ where each $\bm{q}$ is obtained by first drawing $n$ i.i.d.~samples from $\mathcal{N}(\mu_q, \sigma_q)$ and then reordering based on the qualities from high to low. For each sample of qualities, we are able to analytically compute the expected probability that each paper is accepted (and thus the expected conference utility) for the threshold sequential review mechanism and the threshold parallel review mechanism:
\begin{itemize}
    \item For the threshold parallel review mechanism, $\Pr(\text{paper $i$ is accepted}) = 1-G_{\sigma_r}(\tau^\ProbAcc-r_i)$, where $G_{\sigma_r}$ is the c.d.f.~of the zero-mean Gaussian distribution with standard deviation $\sigma_r$.
    \item For the threshold sequential review mechanism, $\Pr(\text{paper $i$ is accepted}) = \prod_{j=1}^{i-1}\left(1-G_{\sigma_r}(\tau^s_{rev}-r_j)\right)\cdot\left(1-G_{\sigma_r}(\tau^s_{acc}-r_i)\right)$.
\end{itemize}
For the isotonic mechanism, due to the non-linear isotonic regression, we do not have a closed-form expected conference utility. Thus, we first draw $100$ review noise terms $\bm{\epsilon}$ and numerically estimate the expected conference utility by taking the average.
Finally, for all implemented mechanisms, the estimated conference utility is the average over \sisetup{group-separator={,}}\num{10000} samples of $\bm{q}$.

For each Gaussian review setting, we further optimize the threshold(s) of the three types of mechanisms using stochastic gradient descent. To estimate the gradient, we simply vary the threshold by a small step and draw \sisetup{group-separator={,}}\num{10000} samples of $\bm{q}$ to estimate the change of conference utility.


\subsubsection{Results}
\label{subsubsec:gaussian-results}

\paragraph{Conference utility. }
We first see how different parameters affect the conference utility achieved by the sequential review mechanism. Let $U_c^p$, $U_c^s$, and $U_c^i$ be the expected conference utilities under the parallel review mechanism, sequential review mechanism, and isotonic mechanism with the optimized thresholds, respectively. We are interested in the relative conference utility of the sequential review mechanism, $\hat{U}_c^s = \frac{U_c^s - U_c^p}{U_c^i - U_c^p}$. That is, the conference utility of the sequential review mechanism while normalizing the baseline $U_c^p$ to 0 and the upper bound $U_c^i$ to 1. 
As shown in Figure~\ref{fig:utility}, the sequential review mechanism exhibits a significant improvement over the baseline, with at least a $40\%$ improvement towards the upper bound across a broad range of parameter settings. Note that the error bars are particularly large for small $n$, large $\mu_q$, and small $\sigma_r$ because in these cases the difference between the performances of three mechanisms is small. Furthermore, we observe that the sequential review mechanism is particularly effective in the following three cases.
\begin{enumerate}
    \item Papers have lower qualities. In this case, due to the review noise, reviewing more papers will increase the probability of accepting more low-quality papers and thus harm the conference quality. Therefore, the sequential review mechanism that rejects lower-quality papers without reviewing them becomes more beneficial compared with the parallel review mechanism.
    \item The author has more papers. Intuitively, as the number of papers increases, more low-quality papers that are mistakenly accepted by the parallel review mechanism are rejected by the sequential review mechanism, which enlarges the gap.
    \item Reviews are noisier. In this case, the sequential review mechanism benefits the conference's utility by utilizing the author's ranking information. At a high level, the noisier the reviews are, the more the conference's decisions rely on the author's ranking information which the parallel review mechanism does not utilize.
\end{enumerate}

\begin{figure*}[htb]
     \centering
     \begin{subfigure}[b]{0.45\textwidth}
         \centering
         \includegraphics[width=\textwidth]{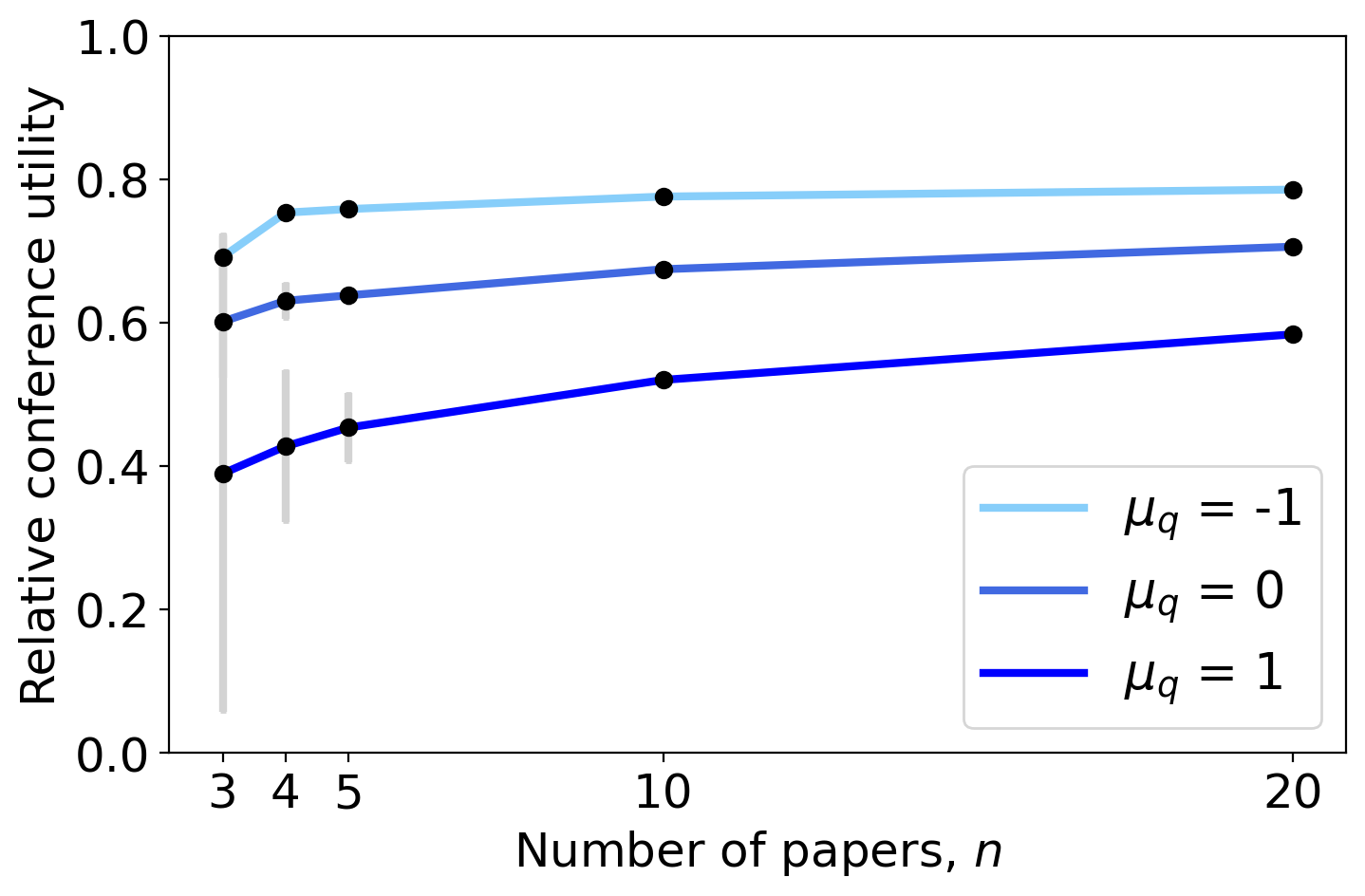}
     \end{subfigure}
     \hfill
     \begin{subfigure}[b]{0.45\textwidth}
         \centering
         \includegraphics[width=\textwidth]{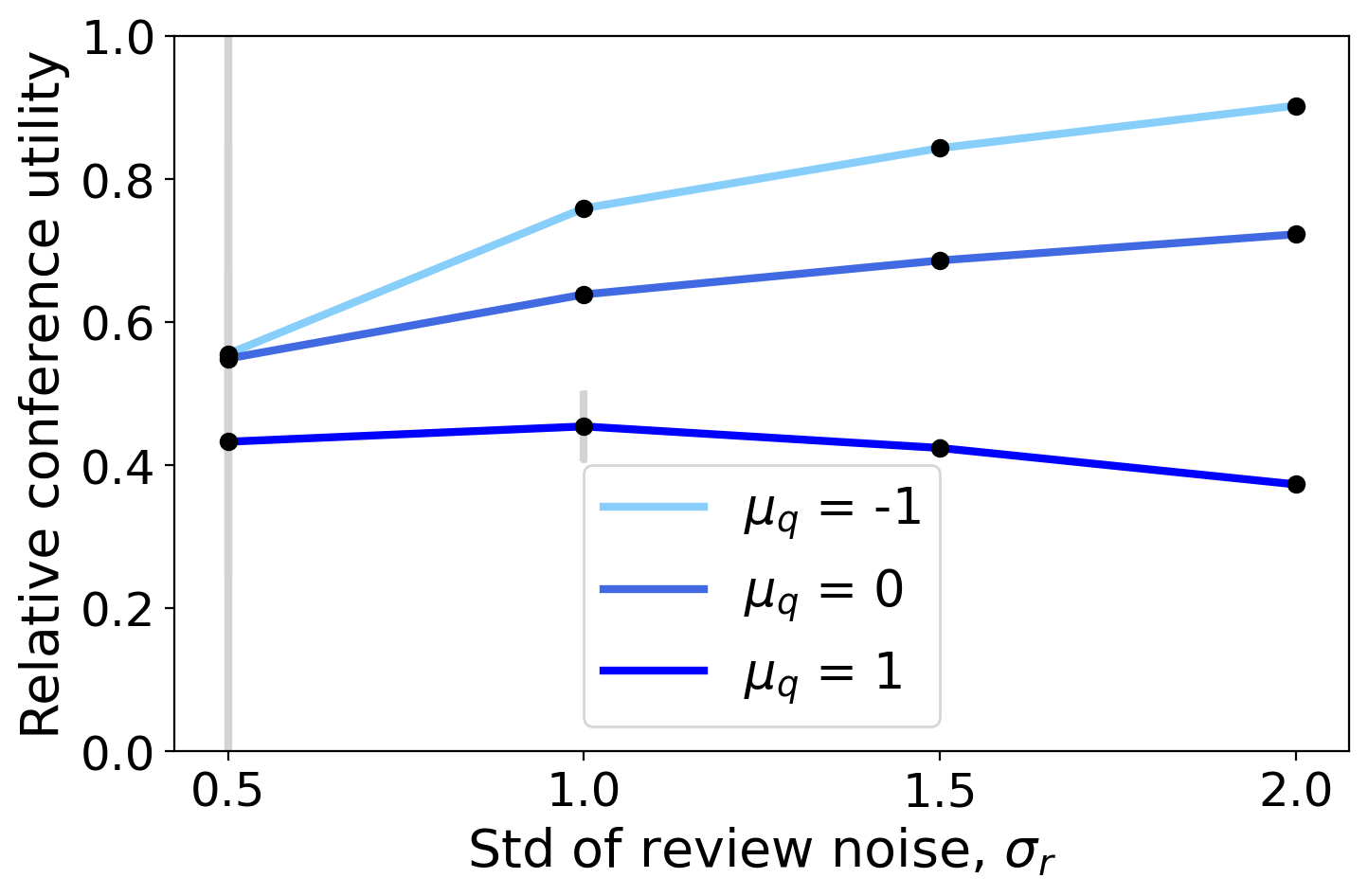}
     \end{subfigure}
     \hfill
\caption{The relative conference utility under different parameter settings. Unless otherwise specified, the default parameter setting is $\varphi^G= (n=5, \mu_q = -1, \sigma_q=2, \sigma_r = 1)$. 
}\label{fig:utility}
\end{figure*}

\paragraph{Review burden. }
Here, we study how many reviews can be saved using a sequential review mechanism. Meanwhile, we study how much improvement can be made in the average quality of the reviewed paper as the review burden decreases.

Suppose an author has $n$ papers for review, we define the \emph{review burden} of a mechanism as the expected number of papers that are reviewed. The review burden of a parallel review mechanism is thus $B^p=n$, while the review burden of a sequential review mechanism $B^s$ is weakly smaller than $n$ depending on the thresholds. We are interested in the relative review burden $\hat{B}^s = \frac{B^s}{B^p}$ conditioned on achieving the same conference utility. By definition, $0<\hat{B}^s\le 1$ and a smaller relative review burden imply that the sequential review mechanism can save more reviews compared with the parallel review mechanism without harming the conference utility. 

\Cref{fig:burden} (a) shows the effectiveness of the sequential review mechanism in reducing the review burden while \cref{fig:burden} (b) presents the ability of the sequential review mechanism to improve the average quality of the reviewed papers. Specifically, we find that the sequential review mechanism is relatively more effective when 1) the author is more likely to write low-quality papers (indicated by smaller $\mu_q$), 2) the author has more papers (indicated by larger $n$) and 3) reviews are noisier (indicated by larger $\sigma_r$).\footnote{Figures that investigating the review noise are deferred to appendix due to space limitation.} The intuition is in line with our discussions on the conference utility. Notably, even when the author has only two papers, the sequential review mechanism can still reduce the review burden by 10-30\% depending on the prior. 

\Cref{fig:burden} (b), we observe that the sequential review mechanism can significantly improve the average quality of the reviewed papers, as low-quality papers that are ranked at the bottom are more likely to be rejected without review.
In reality, the reduction in the review burden and improvement in the quality of reviewed papers may potentially improve the review quality, creating a virtuous circle that ultimately benefits the conference more.  However, this is not explicitly captured by our model, which thus likely underestimates the improvement using a sequential review mechanism.

\begin{figure*}[htb]
     \centering
     \begin{subfigure}[b]{0.45\textwidth}
         \centering
         \includegraphics[width=\textwidth]{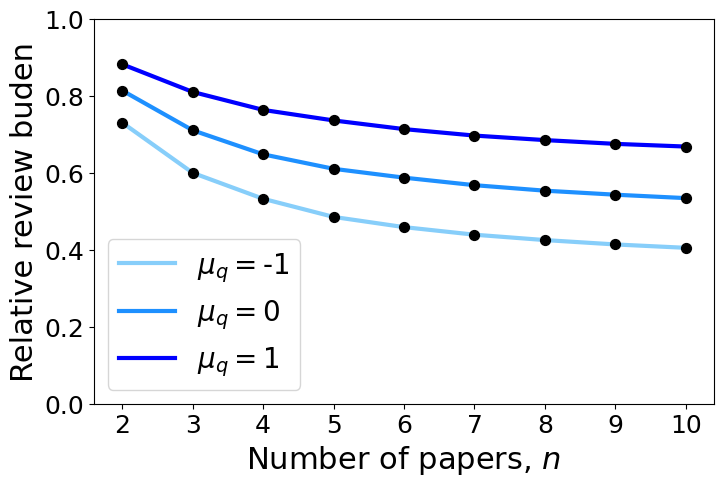}
         \caption{Relative review burden v.s.~$n$.}
     \end{subfigure}
     \hfill
     \begin{subfigure}[b]{0.45\textwidth}
         \centering
         \includegraphics[width=\textwidth]{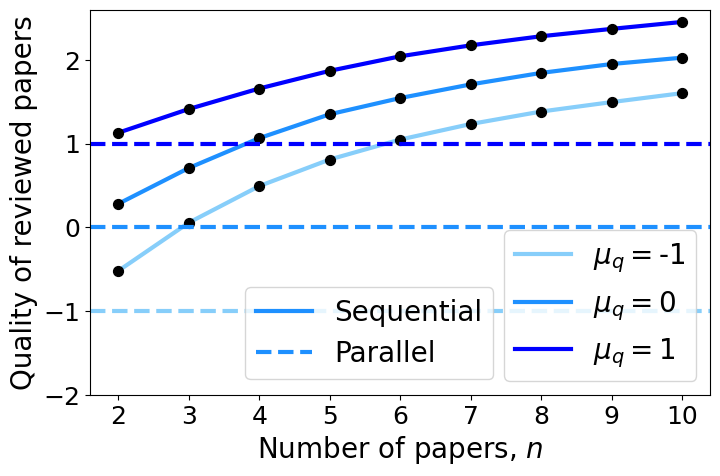}
         \caption{Average quality of the reviewed paper v.s.~$n$.}
     \end{subfigure}
     \hfill
\caption{The relative conference utility and the relative review burden under different parameter settings. Unless otherwise specified, the default parameter setting is $\varphi^G= (n=5, \mu_q = -1, \sigma_q=2, \sigma_r = 1)$.}\label{fig:burden}
\end{figure*}

\subsection{Softmax Review Noise With Real Data}\label{subsec:real_data}
The real review data has two features that are not adequately modeled by the previously discussed Gaussian noise model. First, review scores are integers, not continuous real values. Second, each paper has multiple independent review scores, rather than a single score.   To incorporate these distinctions, we present a more fine-grained model, wherein the review score is characterized by a softmax function. Furthermore, in this section, we optimize the mechanisms for the entire population, rather than an individual agent, based on the empirical quality distribution estimated from real data.

\subsubsection{Model and Experiment Setup}\label{subsec:setup} 
The parameters $(\pi_n, k, \mu_q, \sigma_q, t_r)$ defines a \emph{softmax review setting}.  Authors are ex-ante homogeneous.   The pdf of the distribution of the number of papers an author has is denoted as $\pi_n(\cdot): \mathbb{N}_+\rightarrow [0,1]$.   No matter how many papers $n_i$ author $i$ has, the qualities of the papers are generated by $n_i$ i.i.d.~samples from the same Gaussian paper quality distribution $\mathcal{N}(\mu_q, \sigma_q)$.  Then, conditioned on the quality $q$ of a paper, the conference draws $k$ i.i.d.~review scores from softmax distribution: 
\begin{equation*}
    \Pr(r = s | q) = \frac{e^{-t_r\cdot (s-q)^2}}{\sum_{s'\in \mathcal{S}} e^{-t_r\cdot (s'-q)^2}},
\end{equation*}
where $t_r$ is a temperature parameter which models the noisiness of the reviews, and $\mathcal{S}$, a finite set of integers, is the set of all possible review scores. 

The Monte-Carlo estimation process of this softmax review model is essentially the same as the Gaussian review model. The key difference lies in the thresholds of the review mechanism. Instead of setting thresholds on the review scores directly, the thresholds are set in the posterior world. In particular, conditioned on $k$ i.i.d.~review scores of a paper and the prior of paper quality, the conference can compute an expected quality of the paper. Then, a paper will be accepted (or the review process will continue) if and only if the expected quality surpasses the acceptance (review) threshold. Because of this complexity, we will not include the discussions of the isotonic mechanism, where a reasonably accurate estimate of the performance is computationally difficult.

We use real review datasets to fit the model, where a key challenge is how to deal with the coauthorship problem.  Our solution is to assign every paper to one of its authors and only simulate and use the ranking information from that author. This reduces the multi-author problem to the single-author problem which can be solved by our method. Details of the model-fitting procedure are shown as follows.
\begin{itemize}
    \item 
    First, we want to estimate $\pi_n$. We construct an author-paper mapping from the dataset that maps each author to the set of papers he/she submits. Then, the following procedure is iterated until all papers are assigned with one author:
    \begin{enumerate}[label={\arabic*)}]
        \item Greedily select the author with the largest number of papers from the current author-paper mapping, employing random tie-breaking if necessary.
        \item Assign those papers to the selected author in the previous step.
        \item Remove the author and his/her papers from the author-paper mapping.
    \end{enumerate} 
    This gives us a one-to-one mapping between authors and papers, and $\pi_n$ is set to be the empirical distribution of the number of papers each author has.
    \item While the prevailing choice is now often soliciting four reviews per paper, we set $k=3$ (the second most common choice) due to computational considerations.
    \item We use the average review scores as the true qualities of papers so that $\mu_q$ and $\sigma_q$ can be straightforwardly estimated.
    \item Finally, $t_r$ can be learned using an MAP algorithm \cite{SORENSON198285}.
\end{itemize}

\subsubsection{Datasets}
We use the public ICLR OpenReview datasets from 2021 to 2023. The datasets are reasonably large and include both the accepted and the rejected papers (withdrawn papers are excluded from our analysis), which are suitable to fit our model. Some key features of the datasets and the learned parameters are summarized in \cref{tab:datasets}, where desk-rejected papers are excluded.

\begin{table}[htb]
\centering
\begin{tabular}{|c|c|c|c|c|c|}
\hline
Year & \# of papers & \# of papers per author & $\mu_q$ & $\sigma_q$ & $t_r$ \\ \hline
2021 & 2595         & 1.806                   & 5.511   & 1.009      & 0.513 \\ \hline
2022 & 2612         & 1.834                   & 5.519   & 1.286      & 0.363 \\ \hline
2023 & 3812         & 1.926                   & 5.468   & 1.295      & 0.342 \\ \hline
\end{tabular}
\caption{Key features and learned model parameters of the recent ICLR datasets.}\label{tab:datasets}
\end{table}

Perhaps not surprisingly, both the total number of submitted papers and the average number of papers per author (where papers with coauthors are attributed to the author with the highest paper count) have been steadily increasing over the years. This trend indicates a rise in researchers entering the field and an inclination for authors to produce a larger number of papers. \Cref{fig:empirical-distribution} presents the empirical distribution of the number of papers per author for each year. Notably, the fraction of authors with only one submission decreases over time and the fraction of authors with more than $5$ papers significantly increases in 2023.

\begin{figure}[htbp]
    \centering
    \begin{subfigure}[b]{0.32\textwidth}
        \includegraphics[width=\textwidth]{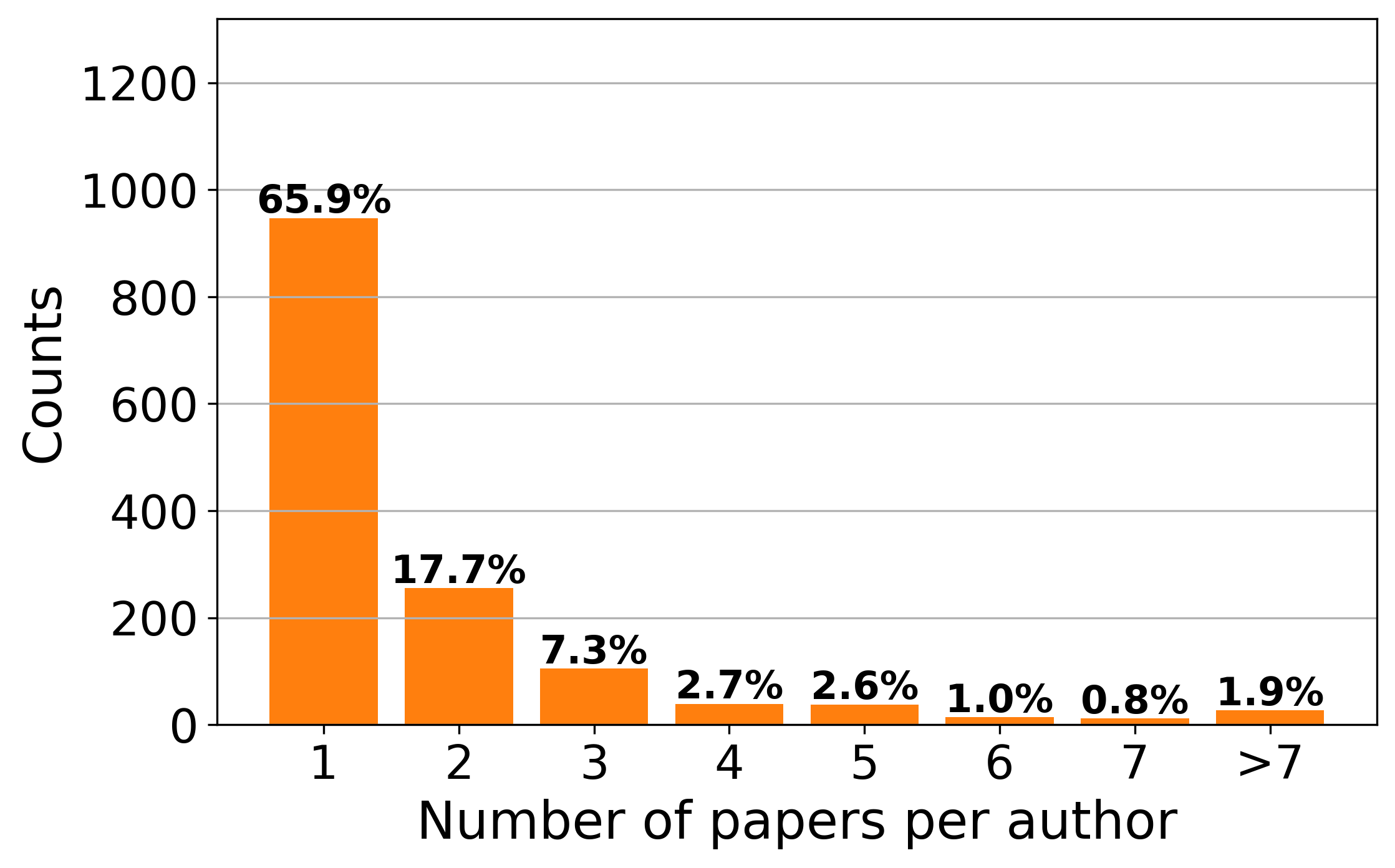}
        \caption{2021}
    \end{subfigure}
    \hfill
    \begin{subfigure}[b]{0.32\textwidth}
        \includegraphics[width=\textwidth]{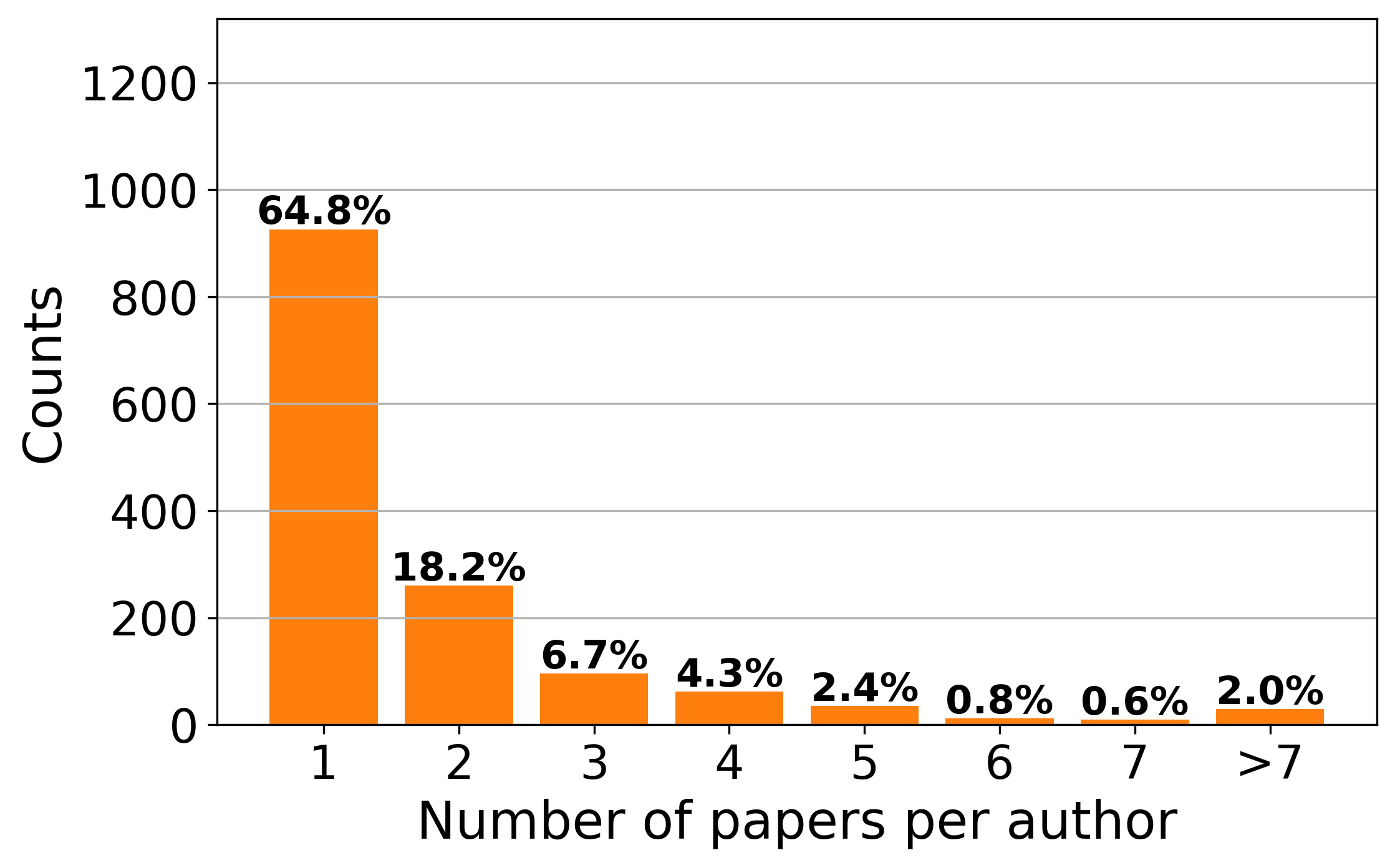}
        \caption{2022}
    \end{subfigure}
    \hfill
    \begin{subfigure}[b]{0.32\textwidth}
        \includegraphics[width=\textwidth]{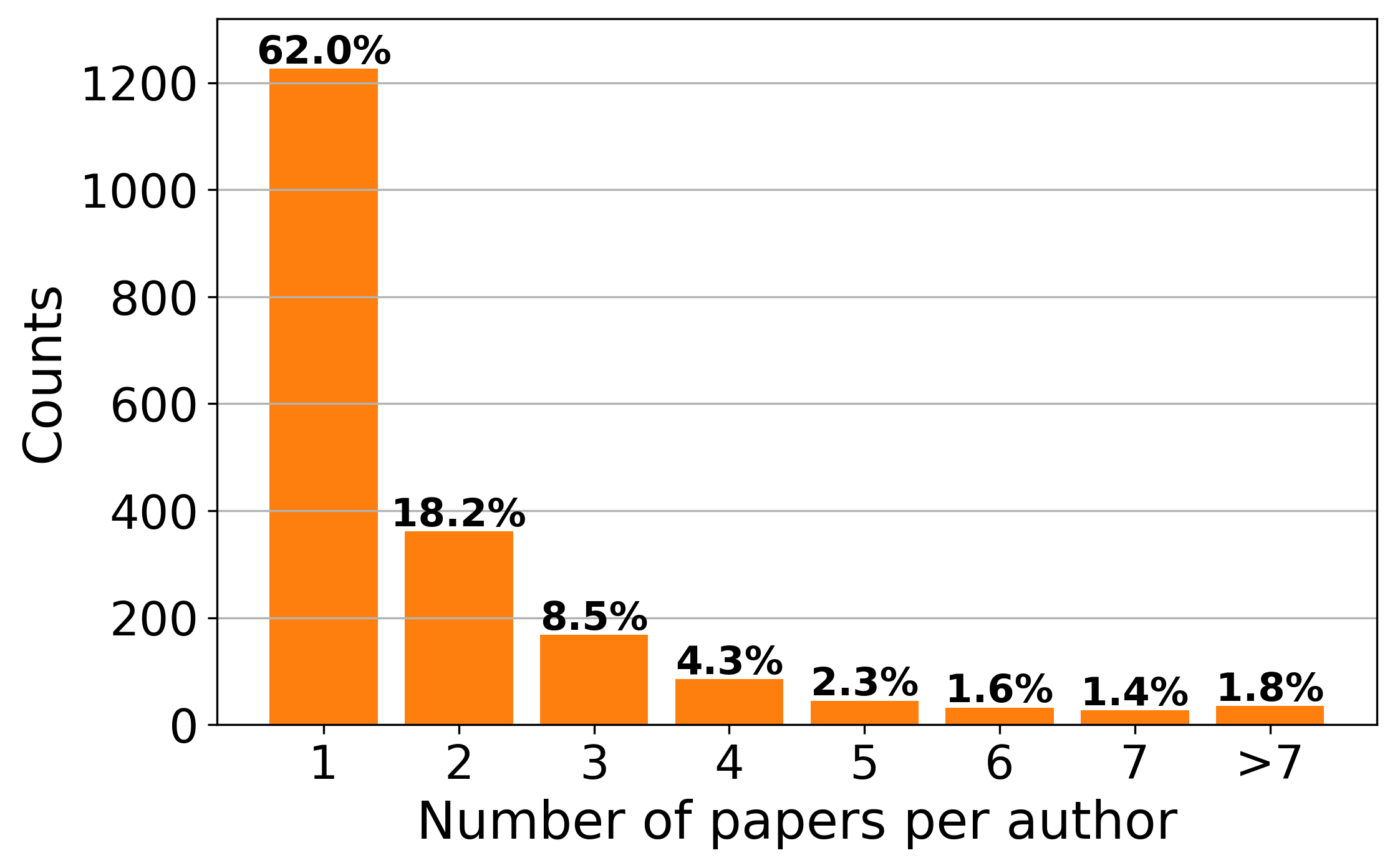}
        \caption{2023}
    \end{subfigure}

    \caption{Empirical distributions of the number of papers each author has for ICLR 2021-2023. Papers with multiple authors are attributed to the author with the largest number of papers.}\label{fig:empirical-distribution}
\end{figure}

\subsubsection{Results}

\Cref{fig:real-burden} illustrates the relative review burden (as defined in \cref{subsubsec:gaussian-results}) of the sequential review mechanism. Recall that the smaller the relative review burden is, the more reviews can be saved by implementing the sequential review mechanism while ensuring an equivalent conference utility to that of the parallel review mechanism. Our results indicate that over 20\% of the review burden can be saved under this real-data estimated model. Moreover, the relative review burden decreases from 2021 to 2023 as more and more authors have more than one paper which enlarges the effect of the sequential review mechanism. 

\begin{figure}[htbp]
    \centering

    \begin{subfigure}[b]{0.5\textwidth}
        \includegraphics[width=\textwidth]{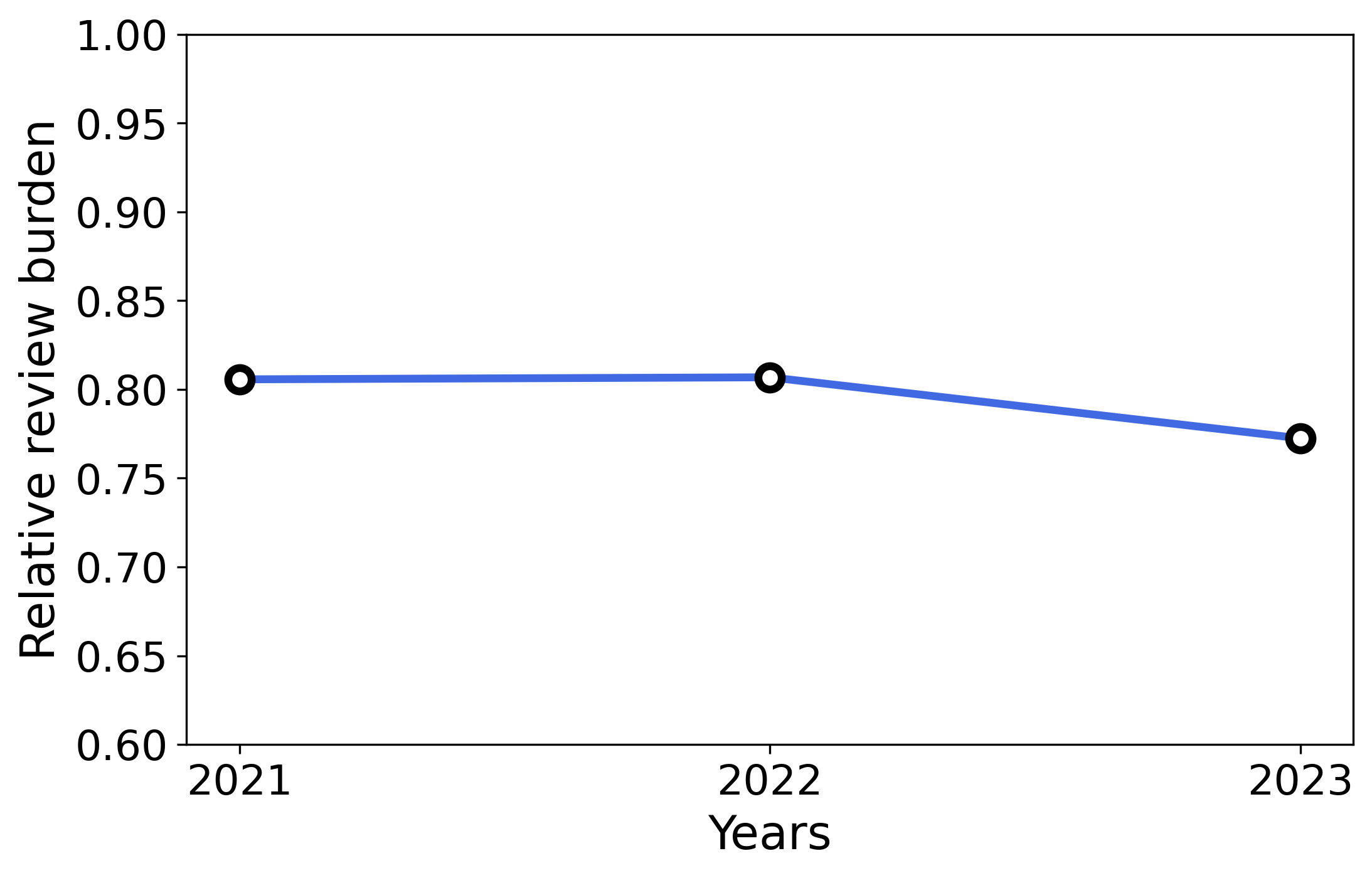}
    \end{subfigure}

    \caption{Relative review burden of the threshold sequential review mechanism for ICLR 2021-2023.}\label{fig:real-burden}
\end{figure}

\section{Endogenous Paper Quality}
\label{sec:endogeneous}
This section considers the setting where authors have the choice of the quality of papers they write. Papers of higher quality have a higher probability of being accepted and bring higher rewards to the author if accepted. However, producing a high-quality paper usually requires greater effort and time from the author. The question of interest is: can we incentivize authors to exert more effort on improving the quality of a smaller set of papers instead of producing more papers each with a lower quality? 

We examine the power of the sequential review mechanism in addressing this problem. Using an economic concept called the \emph{marginal rate of substitution} (MRS), we show that authors always have a stronger willingness to write more high-quality papers and less low-quality papers under the sequential review mechanism compared with the parallel review mechanism.

\subsection{Binary Effort}
As a toy example, we first consider a binary effort setting. Suppose the author can either exert a high effort to write a high-quality paper with an acceptance probability $p_h$ or exert a low effort to write a lower-quality paper with an acceptance probability $p_l<p_h$. Furthermore, the acceptance of a low-quality paper gives the author a reward $u_l$ while the acceptance of a high-quality paper brings a reward $u_h\ge u_l$. Let $U^\mathcal{M}_a(n_h, n_l)$ be the expected reward of writing $n_h$ high-quality papers and $n_l$ low-quality papers under the true ranking $\pi^*$ and mechanism $\mathcal{M}$. We first introduce a key concept in this section.

\begin{definition}
    Given a review mechanism, the marginal rate of substitution between a high-quality paper and a low-quality paper ($MRS_{h,l}$) is defined as the ratio between the expected reward gain of writing one more high-quality paper and the expected reward gain of writing one more low-quality paper. Mathematically, $MRS_{h,l}(n_h, n_l)=\frac{U_a(n_h+1, n_l) - U_a(n_h, n_l)}{U_a(n_h, n_l+1) - U_a(n_h, n_l)}$ where $U_a$ is the expected reward of the author.
\end{definition}

$MRS_{h,l}$ measures the number of low-quality papers that an author is willing to give up for an additional high-quality paper. A higher MRS implies that an additional high-quality paper values more than an additional low-quality paper under the corresponding mechanism. 

\begin{lemma}\label{lem:mrs-binary}
    The naive sequential review mechanism has a weakly higher $MRS_{h,l}$ than the parallel review mechanism, i.e.~$MRS^s_{h,l}\ge MRS^p_{h,l}>1$ for any $n_h, n_l\in \mathbb{N}_0$. Moreover, the statement is strict if $n_l\ge 1$.
\end{lemma}

The proof is deferred to \cref{app:mrs-binary}. Intuitively, the lemma follows because the sequential review mechanism applies an exponential discount on the utilities contributed by the bottom-ranked, low-quality papers. As a result, the ratio between the marginal utility of writing a high-quality paper compared to a low-quality paper is greater under the sequential review mechanism than under the parallel review mechanism. We further note (without formal proof) that this intuition can be straightforwardly generalized to any truthful sequential review mechanism, as long as the review policy $\ProbRev$ is monotone in paper quality.

Now, we show that a higher $MRS_{h,l}$ implies a stronger preference of writing more high-quality papers over low-quality papers. Let $U^s_a(n_h, n_l)$ and $U^p_a(n_h, n_l)$ be the expected utility of writing $n_h$ high-quality papers and $n_l$ low-quality papers under the naive sequential review mechanism and the parallel review mechanism respectively.

\begin{theorem}\label{thm:qual_vs_quant}
    For any $n_h, n_l, n'_h, n'_l\in \mathbb{N}_0$ such that $n'_h>n_h$ and $n'_l>n_l$, if $U_a^p(n'_h, n_l)\ge U_a^p(n_h, n'_l)$, $U_a^s(n'_h, n_l)\ge U_a^s(n_h, n'_l)$. Furthermore, there exist settings where $U_a^s(n'_h, n_l)\ge U_a^s(n_h, n'_l)$, but $U_a^p(n'_h, n_l)< U_a^p(n_h, n'_l)$.
\end{theorem}

In words, \cref{thm:qual_vs_quant} shows that whenever the author wants to write $k_h=n'_h-n_h$ more high-quality papers compared with writing $k_l=n'_l-n_l$ more low-quality papers under the parallel review mechanism, she is always willing to do so under the sequential review mechanism. Furthermore, based on the proof, it is easy to see that the condition is not necessary. That is, there exist some cases where the author wants to write $k_l$ more low-quality papers rather than writing $k_h$ more high-quality papers under the parallel review mechanism, but she is willing to choose writing more high-quality papers under the sequential review mechanism.

\subsection{Finite Effort}

The intuitions and results from the binary effort setting can be straightforwardly generalized to a finite effort setting. Now, suppose the author has $m$ distinct effort levels to choose from. Suppose without loss of generality that these choices of effort result in $m$ distinct acceptance probabilities such that $p_1>p_2>\cdots > p_m$. Let $U_a(\bm{n})$ be the expected reward of writing $n_i$ papers with acceptance probability $p_i$ for every $i\in [m]$ where $\bm{n} = (n_1, n_2,\ldots, n_m)$. We first generalize the definition of the marginal rate of substitution.

\begin{definition}
    Given a review mechanism, the marginal rate of substitution (MRS) between a paper with acceptance probability $p_i$ and a paper with acceptance probability $p_j$ in the finite setting is defined as $MRS_{i,j}(\bm{n})=\frac{U_a(\bm{n}') - U_a(\bm{n})}{U_a(\bm{n}'') - U_a(\bm{n})}$ where $\bm{n}' = (n_1,\ldots, n_{i-1}, n_i+1, \ldots, n_m)$ and $\bm{n}'' = (n_1,\ldots, n_{j-1}, n_j+1, \ldots, n_m)$.
\end{definition}

Intuitively, $MSR_{i,j}$ is the ratio between the gain of the reward of writing one more paper with acceptance probability $p_i$ and the gain of the reward of writing one more paper with acceptance probability $p_j$. Again, a higher $MSR_{i,j}$ implies that the author is willing to give up more papers with acceptance probability $p_j$ to write one more paper with probability $p_i$. The following lemma generalizes \cref{lem:mrs-binary} to the finite effort setting. The proof is left in \cref{app:finite_effort}.

\begin{lemma}\label{lem:mrs-finite}
    The naive sequential review mechanism has a weakly higher $MRS_{i,j}(\bm{n})$ than the parallel review mechanism for any $\bm{n}\in \mathbb{N}_0^{|\bm{n}|}$ and $i< j$. Moreover, the statement is strict if $\sum_{k = i+1}^m n_k > 0$.
\end{lemma}

With the help of \cref{lem:mrs-finite}, we generalize the main result, i.e.~\cref{thm:qual_vs_quant}, from the binary effort setting to the finite effort setting. We defer the proof to \cref{app:finite_effort}.

\begin{theorem}\label{thm:qual_vs_quant_finite}
    For any $1\le i<j\le m$, let $n_i, n_j, n'_i, n'_j\in \mathbb{N}_0$ such that $n'_i>n_i$ and $n'_j>n_j$. Then, if $U_a^p(n'_i, n_j)\ge U_a^p(n_i, n'_j)$, $U_a^s(n'_i, n_j)\ge U_a^s(n_i, n'_j)$. Furthermore, there exist settings where $U_a^s(n'_i, n_j)\ge U_a^s(n_i, n'_j)$, but $U_a^p(n'_i, n_j)< U_a^p(n_i, n'_j)$.
\end{theorem}

\section{Limitations, Discussions, and Future Work}
\label{sec:discussion}

Here, we discuss the limitations of our analysis and how to possibly implement the proposed method in practice.  However, the implementation details can likely be further improved by future work.  

\paragraph{Coauthorship}
The first limitation of our model is the assumption that every paper has only one author. 
However, in practice, coauthorship is an inevitable issue for the implementation of our method. A straightforward solution is to assign each paper to one of its authors and only solicit the ranking information from that author. 

Various strategies can be employed for the paper-to-author assignment. For example, we can assign each paper only to its first author, driven by the notion that the first author may possess the most accurate insight into the paper's quality. However, this assignment weakens the power of the sequential review mechanism, as many first authors have fewer than three submissions to their name. Another approach is to assign each paper to the author with the largest number of submissions, as discussed in \cref{subsec:setup}. In a recent work \cite{wu2023isotonic} that generalizes the isotonic mechanism to the multi-author setting, it is shown that this greedy assignment is truthful and has appealing robust approximation guarantees for the isotonic mechanism.

\paragraph{The trade-off between review burden and delay}

Another practical concern is the delay in the review process. As presented, the mechanism requires $n$ phases where $n$ is the maximum number of papers owned by any single author. Within phase $i$, the $i$'th ranked paper of every author (if any) is simultaneously reviewed and an acceptance/rejection decision is promptly reached. Then, phase $i + 1$ can only start after the reviews of phase $i$ are in. 
Intuitively, this implementation significantly delays the review process but minimizes the review burden. 
At the other extreme, the conference could only have one review phase, where all the papers are simultaneously reviewed. Then, for each author, the acceptance/rejection decision is made in sequence based on the sequential review mechanism. Note that any paper that should not have been reviewed according to the mechanism could be rejected, its reviews simply ignored.  Here the review burden is not decreased at all (though, as we saw, the conference quality may still be increased).

Intuitively, varying the number of review phases can trade off the review burden and the delay in the review process. 
In reality, perhaps a two-phase implementation of the sequential review mechanism can achieve a desirable trade-off. For example, some conferences such as AAAI and EC are already implementing a two-phase review mechanism: all papers are assigned with two reviews in the first phase and only papers with at least one good review will enter the second phase where two more reviews are assigned. We can integrate the sequential review mechanism with this framework. In the first phase, the top $\max(3,\lfloor n_i/2\rfloor)$ ranked papers of each author with $n_i$ papers are assigned with two reviews, and the remaining papers are assigned with one review (so that no paper is rejected without reviewing). Then, any paper with two negative reviews and those papers that are ranked lower than them by authors are rejected with no further review. The surviving papers enter the second phase, wherein they are assigned additional reviews to reach a total of four reviews each, and a sequential review mechanism is implemented. The advantage of using the sequential review mechanism to assist peer review is that the author's information can be leveraged to prioritize the reviewing of high-quality papers.

\paragraph{Broader Applications}
Although our paper primarily focuses on the conference peer review problem, our insights can be potentially applied to address general principal-agent problems with decisions that rely on noisy evaluations. For example, consider content-recommendation platforms where content producers are entrusted with the ability to highlight their videos or products for priority recommendations. While the original intent is to leverage producer signals to promote high-quality content, strategic producers may game this mechanism by prioritizing utility-maximizing content, such as advertisements, which may not necessarily align with user preferences. Our proposed sequential review mechanism offers a solution to this incentive issue by soliciting content rankings from producers and penalizing lower-ranked items with reduced user recommendations, particularly if the top-ranked items receive unfavorable user feedback.
Additional applications suitable for our method include employee recruitment and second-hand product trading markets, where principals can benefit from agents' ranking information.

\section{Conclusion and Future Work}

In the setting of (conference) peer review, we study the problem of how to elicit honest information from authors, who themselves are interested in the outcome. Our main contribution is a framework for designing mechanisms capable of eliciting quality rankings from authors with multiple submissions. Compared with the previous isotonic mechanism, our mechanism works within a more realistic utility model for peer review and addresses a key incentive issue that plagued the previous method. We further investigate the advantages of our mechanism from the aspects of reducing reviewing workload, improving the average quality of the reviewed papers, and incentivizing authors to focus more on the quality of papers rather than the quantity.

The noisy conference peer review has been the source of frequent complaints and concerns for authors over the years. We believe that properly utilizing author information has the potential to improve conference decision-making and mitigate the dilemma of peer review. 
Nonetheless, our work leaves room for further exploration. The largest space of this is in optimally integrating the sequential review mechanism with a practical multi-phase review mechanism. In particular, in order to optimally trade off the conference utility, review burden, and some definition of fairness, how many review phases should the conference have? Furthermore, within each review phase, how many papers from each author should be reviewed and how many reviews should be solicited for each paper?

\section{Acknowledgements}
The authors would like to thank all the anonymous reviewers whose insightful comments and constructive feedback have contributed to the improvement of this work. This work is supported by the National Science Foundation under Grants \#2007256 and \#2313137.

\bibliographystyle{plainnat}
\bibliography{main}

\begin{thebibliography}{32}
\providecommand{\natexlab}[1]{#1}
\providecommand{\url}[1]{\texttt{#1}}
\expandafter\ifx\csname urlstyle\endcsname\relax
  \providecommand{\doi}[1]{doi: #1}\else
  \providecommand{\doi}{doi: \begingroup \urlstyle{rm}\Url}\fi

\bibitem[Aziz et~al.(2019)Aziz, Lev, Mattei, Rosenschein, and Walsh]{aziz2019strategyproof}
Haris Aziz, Omer Lev, Nicholas Mattei, Jeffrey~S Rosenschein, and Toby Walsh.
\newblock Strategyproof peer selection using randomization, partitioning, and apportionment.
\newblock \emph{Artificial Intelligence}, 275:\penalty0 295--309, 2019.

\bibitem[Bazi(2020)]{bazi2020peer}
Tony Bazi.
\newblock Peer review: single-blind, double-blind, or all the way-blind?
\newblock \emph{International Urogynecology Journal}, 31\penalty0 (3):\penalty0 481--483, 2020.

\bibitem[Blank(1991)]{blank1991effects}
Rebecca~M Blank.
\newblock The effects of double-blind versus single-blind reviewing: Experimental evidence from the american economic review.
\newblock \emph{The American Economic Review}, pages 1041--1067, 1991.

\bibitem[Cabanac and Preuss(2013)]{asi.22747}
Guillaume Cabanac and Thomas Preuss.
\newblock Capitalizing on order effects in the bids of peer-reviewed conferences to secure reviews by expert referees.
\newblock \emph{Journal of the American Society for Information Science and Technology}, 64\penalty0 (2):\penalty0 405--415, 2013.
\newblock \doi{https://doi.org/10.1002/asi.22747}.
\newblock URL \url{https://onlinelibrary.wiley.com/doi/abs/10.1002/asi.22747}.

\bibitem[Cohen et~al.(2016)Cohen, Pattanaik, Kumar, Bies, De~Boer, Ferro, Gilchrist, Isbister, Ross, and Webb]{cohen2016organised}
Adam Cohen, Smita Pattanaik, Praveen Kumar, Robert~R Bies, Anthonius De~Boer, Albert Ferro, Annette Gilchrist, Geoffrey~K Isbister, Sarah Ross, and Andrew~J Webb.
\newblock Organised crime against the academic peer review system.
\newblock \emph{British Journal of Clinical Pharmacology}, 81\penalty0 (6):\penalty0 1012, 2016.

\bibitem[Cortes and Lawrence(2021)]{cortes2021inconsistency}
Corinna Cortes and Neil~D. Lawrence.
\newblock Inconsistency in conference peer review: Revisiting the 2014 neurips experiment, 2021.

\bibitem[De~Clippel et~al.(2008)De~Clippel, Moulin, and Tideman]{de2008impartial}
Geoffroy De~Clippel, Herve Moulin, and Nicolaus Tideman.
\newblock Impartial division of a dollar.
\newblock \emph{Journal of Economic Theory}, 139\penalty0 (1):\penalty0 176--191, 2008.

\bibitem[Dhull et~al.(2022)Dhull, Jecmen, Kothari, and Shah]{dhull2022strategyproofing}
Komal Dhull, Steven Jecmen, Pravesh Kothari, and Nihar~B Shah.
\newblock Strategyproofing peer assessment via partitioning: The price in terms of evaluators’ expertise.
\newblock In \emph{Proceedings of the AAAI Conference on Human Computation and Crowdsourcing}, volume~10, pages 53--63, 2022.

\bibitem[Fanelli(2009)]{fanelli2009many}
Daniele Fanelli.
\newblock How many scientists fabricate and falsify research? a systematic review and meta-analysis of survey data.
\newblock \emph{PloS one}, 4\penalty0 (5):\penalty0 e5738, 2009.

\bibitem[Haffar et~al.(2019)Haffar, Bazerbachi, and Murad]{haffar2019peer}
Samir Haffar, Fateh Bazerbachi, and M~Hassan Murad.
\newblock Peer review bias: a critical review.
\newblock In \emph{Mayo Clinic Proceedings}, volume~94, pages 670--676. Elsevier, 2019.

\bibitem[Jecmen et~al.(2020)Jecmen, Zhang, Liu, Shah, Conitzer, and Fang]{jecmen2020mitigating}
Steven Jecmen, Hanrui Zhang, Ryan Liu, Nihar Shah, Vincent Conitzer, and Fei Fang.
\newblock Mitigating manipulation in peer review via randomized reviewer assignments.
\newblock \emph{Advances in Neural Information Processing Systems}, 33:\penalty0 12533--12545, 2020.

\bibitem[Lane et~al.(2022)Lane, Teplitskiy, Gray, Ranu, Menietti, Guinan, and Lakhani]{lane2022conservatism}
Jacqueline~N Lane, Misha Teplitskiy, Gary Gray, Hardeep Ranu, Michael Menietti, Eva~C Guinan, and Karim~R Lakhani.
\newblock Conservatism gets funded? a field experiment on the role of negative information in novel project evaluation.
\newblock \emph{Management science}, 68\penalty0 (6):\penalty0 4478--4495, 2022.

\bibitem[Lawrence and Cortes(2014)]{lawrence2014nips}
Neil Lawrence and Corinna Cortes.
\newblock The nips experiment.
\newblock \emph{See http://inverseprobability. com/2014/12/16/the-nips-experiment (accessed 3 March 2021)}, 2014.

\bibitem[Lee et~al.(2013)Lee, Sugimoto, Zhang, and Cronin]{lee2013bias}
Carole~J Lee, Cassidy~R Sugimoto, Guo Zhang, and Blaise Cronin.
\newblock Bias in peer review.
\newblock \emph{Journal of the American Society for information Science and Technology}, 64\penalty0 (1):\penalty0 2--17, 2013.

\bibitem[Littman(2021)]{littman2021collusion}
Michael~L Littman.
\newblock Collusion rings threaten the integrity of computer science research.
\newblock \emph{Communications of the ACM}, 64\penalty0 (6):\penalty0 43--44, 2021.

\bibitem[Meir et~al.(2021)Meir, Lang, Lesca, Mattei, and Kaminsky]{meir2021market}
Reshef Meir, J{\'e}r{\^o}me Lang, Julien Lesca, Nicholas Mattei, and Natan Kaminsky.
\newblock A market-inspired bidding scheme for peer review paper assignment.
\newblock In \emph{Proceedings of the AAAI Conference on Artificial Intelligence}, volume~35, pages 4776--4784, 2021.

\bibitem[Naghizadeh and Liu(2013)]{naghizadeh2013incentives}
Parinaz Naghizadeh and Mingyan Liu.
\newblock Incentives, quality, and risks: A look into the nsf proposal review pilot, 2013.

\bibitem[Naghizadeh and Liu(2016)]{6920096}
Parinaz Naghizadeh and Mingyan Liu.
\newblock Perceptions and truth: A mechanism design approach to crowd-sourcing reputation.
\newblock \emph{IEEE/ACM Transactions on Networking}, 24\penalty0 (1):\penalty0 163--176, 2016.
\newblock \doi{10.1109/TNET.2014.2359767}.

\bibitem[Sculley et~al.(2018)Sculley, Snoek, and Wiltschko]{sculley2018avoiding}
D~Sculley, Jasper Snoek, and Alex Wiltschko.
\newblock Avoiding a tragedy of the commons in the peer review process, 2018.

\bibitem[Shah(2019)]{Shah2019PrincipledMT}
Nihar~B. Shah.
\newblock Principled methods to improve peer review.
\newblock 2019.
\newblock URL \url{https://api.semanticscholar.org/CorpusID:208781231}.

\bibitem[Shah(2022)]{shah2022challenges}
Nihar~B Shah.
\newblock Challenges, experiments, and computational solutions in peer review.
\newblock \emph{Communications of the ACM}, 65\penalty0 (6):\penalty0 76--87, 2022.

\bibitem[Siegelman(1991)]{radiology.178.3.1994394}
S~S Siegelman.
\newblock Assassins and zealots: variations in peer review. special report.
\newblock \emph{Radiology}, 178\penalty0 (3):\penalty0 637--642, 1991.
\newblock \doi{10.1148/radiology.178.3.1994394}.
\newblock URL \url{https://doi.org/10.1148/radiology.178.3.1994394}.
\newblock PMID: 1994394.

\bibitem[Snodgrass(2006)]{snodgrass2006single}
Richard Snodgrass.
\newblock Single-versus double-blind reviewing: An analysis of the literature.
\newblock \emph{ACM Sigmod Record}, 35\penalty0 (3):\penalty0 8--21, 2006.

\bibitem[Sorenson(1982)]{SORENSON198285}
H.W. Sorenson.
\newblock Parameter and state estimation: Introduction and interrelation.
\newblock \emph{IFAC Proceedings Volumes}, 15\penalty0 (4):\penalty0 85--89, 1982.
\newblock ISSN 1474-6670.
\newblock \doi{https://doi.org/10.1016/S1474-6670(17)62968-9}.
\newblock URL \url{https://www.sciencedirect.com/science/article/pii/S1474667017629689}.
\newblock 6th IFAC Symposium on Identification and System Parameter Estimation, Washington USA, 7-11 June.

\bibitem[Spalvieri et~al.(2014)Spalvieri, Mandelli, Magarini, and Bianchi]{spalvieri2014weighting}
Arnaldo Spalvieri, Silvio Mandelli, Maurizio Magarini, and Giuseppe Bianchi.
\newblock Weighting peer reviewers.
\newblock In \emph{2014 Twelfth Annual International Conference on Privacy, Security and Trust}, pages 414--419. IEEE, 2014.

\bibitem[Srinivasan and Morgenstern(2021)]{srinivasan2021auctions}
Siddarth Srinivasan and Jamie Morgenstern.
\newblock Auctions and prediction markets for scientific peer review.
\newblock \emph{arXiv preprint arXiv:2109.00923}, 2021.

\bibitem[Su(2021)]{su2021you}
Weijie~J Su.
\newblock You are the best reviewer of your own papers: An owner-assisted scoring mechanism.
\newblock In A.~Beygelzimer, Y.~Dauphin, P.~Liang, and J.~Wortman Vaughan, editors, \emph{Advances in Neural Information Processing Systems}, 2021.
\newblock URL \url{https://openreview.net/forum?id=xmx5rE9QP7R}.

\bibitem[Wang and Shah(2018)]{wang20182}
Jingyan Wang and Nihar~B. Shah.
\newblock Your 2 is my 1, your 3 is my 9: Handling arbitrary miscalibrations in ratings, 2018.

\bibitem[Wu et~al.(2023)Wu, Xu, Guo, and Su]{wu2023isotonic}
Jibang Wu, Haifeng Xu, Yifan Guo, and Weijie Su.
\newblock An isotonic mechanism for overlapping ownership, 2023.

\bibitem[Xu et~al.(2018)Xu, Zhao, Shi, Zhang, and Shah]{xu2018strategyproof}
Yichong Xu, Han Zhao, Xiaofei Shi, Jeremy Zhang, and Nihar~B Shah.
\newblock On strategyproof conference peer review.
\newblock \emph{arXiv preprint arXiv:1806.06266}, 2018.

\bibitem[Yan et~al.(2023)Yan, Su, and Fan]{yan2023isotonic}
Yuling Yan, Weijie~J. Su, and Jianqing Fan.
\newblock The isotonic mechanism for exponential family estimation, 2023.

\bibitem[Zhang et~al.(2022)Zhang, Yu, Schoenebeck, and Kempe]{10.1145/3490486.3538235}
Yichi Zhang, Fang-Yi Yu, Grant Schoenebeck, and David Kempe.
\newblock A system-level analysis of conference peer review.
\newblock In \emph{Proceedings of the 23rd ACM Conference on Economics and Computation}, EC '22, page 1041–1080, New York, NY, USA, 2022. Association for Computing Machinery.
\newblock ISBN 9781450391504.
\newblock \doi{10.1145/3490486.3538235}.
\newblock URL \url{https://doi.org/10.1145/3490486.3538235}.

\end{thebibliography}
\newpage
\appendix

\section{The Isotonic Mechanism}
\label{app:isotonic}

For self-containedness, we briefly introduce the isotonic mechanism and why it is not truthful when authors' utility is not convex of the review score. 

Consider the model in \cref{sec:model}. The isotonic mechanism takes the input of a sequence of review scores $\bm{r}$ and a ranking $\pi$ reported by the author, and outputs the modified review scores $\hat{\bm{r}}$ by solving the following convex program:

\begin{align*}
    \min_{\hat{\bm{r}}} \quad &\frac{1}{2}||\hat{\bm{r}} - \bm{r}||^2\\
    s.t.\quad &\hat{r}_{\pi(1)}\ge \hat{r}_{\pi(2)} \ge \cdots \ge \hat{r}_{\pi(n)}.
\end{align*}
That is, the isotonic mechanism projects the raw review scores onto the feasible region $\{\bm{r}: \hat{r}_{\pi(1)}\ge \hat{r}_{\pi(2)} \ge \cdots \ge \hat{r}_{\pi(n)}\}$ and forms the new review scores $\hat{\bm{r}}$. 

Now, we provide an example of why the isotonic mechanism experiences incentive issues when the convex utility assumption is violated.

\begin{example}
    Suppose an author has three papers $a, b$ and $c$, whose qualities are $q_a = 2, q_b = -1, q_c = -1$. The author aims to maximize the expected number of accepted papers. Suppose the review score is rather accurate such that $r_i = q_i$ for $i\in [a,b,c]$, and the conference has an acceptance threshold at $0$. Therefore, if the author reports truthfully, the modified review score $\hat{r}_i$ agrees with the true quality, resulting in the rejection of both paper $b$ and paper $c$. However, if the author ranks paper $a$ in third place, resulting in scores $\hat{r}_a = \hat{r}_b = \hat{r}_c =0$, which means all three papers will be accepted.   This is part of a more general phenomenon where a high paper can push up the ratings of lower papers that are ranked ahead of it, and extends to the case where review scores are noisy.  Furthermore, note that in this example the conference is even better utilizing the parallel review mechanism, resulting a utility of $2$.
\end{example}

\section{Proofs in \cref{sec:seq_mechanism}}

\subsection{Truthfulness of the Memoryless Coin-flip Mechanism}
\label{app:CF_proof}

\begin{proof}[Proof of \cref{thm:cf-truth}]
    By \cref{thm:sufficient_cond}, it is sufficient to show that $\ProbRev^{cf}$ and $\bm{\mu}^{cf}$ are monotone. It is easy to see that $\ProbRev^{cf}$ is monotone based on the construction. Specifically, $\ProbRev^{cf}((\alpha, \gamma)) = 1 > \ProbRev^{cf}(\omega) = 0$ for any $(\alpha, \gamma)\neq \omega$. 

    Next, we show that $\bm{\mu}^{cf}$ is monotone. This includes three steps. For simplicity, we omit the superscript within this proof.
    \begin{enumerate}
        \item We first show that $\bm{\mu}$ is monotone in score. Consider any review state $(\alpha_i, \gamma_i)\neq \omega$ and any pair of review scores $r'\ge r$. Because any state other than $\omega$ has the same order, it is sufficient to show that $\mu(\omega|(\alpha_i,\gamma_i), r')\le \mu(\omega|(\alpha_i,\gamma_i), r)$. This holds because when $\ProbAcc$ is monotone and $\rho$ is increasing, $\mu(\omega|(\alpha_i,\gamma_i), r') = (1-\ProbAcc(r')) (1-\rho(r'))\le (1-\ProbAcc(r)) ((1-\rho(r)) = \mu(\omega|(\alpha_i,\gamma_i), r)$.
        \item Then, we show that $\bm{\mu}$ is monotone in state. Fix a review score $r$ and consider any pair of states $(\alpha',\gamma')\succeq (\alpha,\gamma)$.  We want to show that $\mu(\omega|(\alpha',\gamma'), r)\le \mu(\omega|(\alpha,\gamma), r)$. There are two cases. First, if $(\alpha,\gamma)=\omega$, the statement trivially holds. Second, if $(\alpha',\gamma')\sim (\alpha,\gamma)\succ \omega$, $\mu((\alpha',\gamma'), r) = \mu((\alpha,\gamma), r)$ because the review process is memoryless. In any case, conditioned on the same score, a better state leads to a dominating state distribution in the next round.
        \item Finally, we show that $\bm{\mu}$ is monotone of ordering. Let $\tilde{\mu}_\omega(r_1, r_2|\phi)$ be the probability of reaching the termination state $\omega$ in round $i+2$, given the review state to be $\phi$ in round $i$ and review scores to be $r_1$ and $r_2$ in round $i$ and $i+1$ respectively. Now, consider any review state $\phi\neq \omega$. The following property of the memoryless coin-flip mechanism is the key to the proof.
        \begin{lemma}\label{lem:exchange}
            For any $\phi$ and $r'\ge r$, $\tilde{\mu}^{cf}_\omega(r', r|\phi) = \tilde{\mu}^{cf}_\omega(r, r'|\phi)$.
        \end{lemma}
       The lemma suggests that the state distribution after flipping the ordering of two adjacent review scores remains the same under the memoryless coin-flip mechanism
        We defer the proof of the lemma and present the proof of the theorem first, which has two parts.
        \begin{enumerate}
            \item By \cref{lem:exchange}, $\tilde{\mu}_\omega(r', r|\phi) = \tilde{\mu}_\omega(r, r'|\phi)$. 
            \item Consider four review scores $r_1\le r_2\le r_3\le r_4$. Let $X\sim \tilde{\mu}(r_4, r_1|\phi)$, $Y\sim \tilde{\mu}(r_2, r_3|\phi)$, $X'\sim\tilde{\mu}(r_1, r_4|\phi)$ and $Y'\sim\tilde{\mu}(r_3, r_2|\phi)$. Let $\bar{Z} = \max(X, Y)$, $\underbar{Z} = \min(X, Y)$, $\bar{Z}' = \max(X', Y')$, $\underbar{Z}' = \min(X', Y')$. By \cref{lem:exchange}, $\bar{Z}$ and $\bar{Z}'$ have the same distribution, as do $\underbar{Z}$ and $\underbar{Z}'$. Therefore, $(\bar{Z}, \underbar{Z})$ has the same distribution as $(\bar{Z}', \underbar{Z}')$.
        \end{enumerate}
    \end{enumerate}
\end{proof}

Now, we prove the Lemma.
    \begin{proof}[Proof of \cref{lem:exchange}]
    First note that if $\phi = \omega$, the lemma trivially holds.
    Let $A(r) = (1-\ProbAcc(r)) (1-\rho(r))$ be the probability of termination in the round where the state is not $\omega$ and the paper has a review score $r$.
    \begin{align*}
        \tilde{\mu}^{cf}_\omega(r', r|\phi) &= A(r') + (1-A(r'))A(r)\\
        &= A(r) + (1-A(r))A(r')\\
        & = \tilde{\mu}^{cf}_\omega(r, r'|\phi).
    \end{align*}
    \end{proof}

\subsection{Truthfulness of the Credit Pool Mechanism}
\label{app:CP_proof}

\begin{proof}[Proof of \cref{thm:cp_truth}]
By \cref{thm:sufficient_cond}, it is sufficient to show that $\ProbRev^{cp}$ and $\bm{\mu}^{cp}$ are monotone. 
Under the definition of the ordering of states, it is clear that the review policy is monotone. Because for any $B'\ge B$, it follows that whenever $\ProbRev^{cp}(B)=1$, $\ProbRev^{cp}(B')=1$.

Now, we show that the transition probability mapping is monotone. For simplicity, we omit the superscript.
\begin{enumerate}
    \item We first show that $\bm{\mu}$ is monotone in score. This is true because for any $B_i\ge 0$ and any $r'\ge r$, $B'_{i+1} = B_{i} + \beta(r')\ge B_{i} + \beta(r) = B_{i+1}$. Furthermore, if $B_i<0$, $B'_{i+1} = B_{i+1} = \omega$. In both cases, a better review score leads to a (weakly) preferred state.
    \item Then, we show that $\bm{\mu}$ is monotone in state. Fixing a review score $r$, a better review state in round $i$ indicates a better review state in round $i+1$ because 1) if $B'_{i}\ge B_{i}\ge 0$, we have $B'_{i+1} = B'_{i} + \beta(r)\ge B_{i} + \beta(r) = B_{i+1}$; 2) if $B_{i} < 0$, it follows that $B'_{i+1}\succeq B_{i+1} = B_{i}$ for any $B'_{i+1}$.
    \item Finally, we show that $\bm{\mu}$ is monotone in ordering, which has the following two steps.
    \begin{enumerate}
        \item First, we show that if two review scores in adjacent rounds are swapped, it is always better to put the higher score ahead. Consider any state $B\ge$ and any pair of review scores $r'\ge r$. Because the transition mapping is deterministic, if $B+\beta(r)<0$, we have $\tilde{\mu}(r, r'|B) =\omega$ with probability 1, and thus is trivially dominated by $\tilde{\mu}(r', r|B)$. If $B+\beta(r)>0$, $\tilde{\mu}(r', r|B) = \tilde{\mu}(r, r'|B) = B + \beta(r') + \beta(r)$ with probability 1. In any case, ranking the higher review score ahead results in a better review state.
        \item Suppose $r_1\le r_2\le r_3\le r_4$ and $r_1+r_4 = r_2+r_3$. For any review state $B_i$ in round $i$, we want to reason about the review states after two rounds under different pairs of review scores. Let $\bar{Z} = \max(\tilde{\mu}(r_3,r_2|B), \tilde{\mu}(r_1,r_4|B))$, $\underbar{Z} = \min(\tilde{\mu}(r_3,r_2|B), \tilde{\mu}(r_1,r_4|B))$, $\bar{Z}' = \max(\tilde{\mu}(r_4,r_1|B), \tilde{\mu}(r_2,r_3|B))$ and $\underbar{Z}' = \min(\tilde{\mu}(r_4,r_1|B), \tilde{\mu}(r_2,r_3|B))$. We want to show that $\bar{Z}'\ge \bar{Z}$ and $\underbar{Z}'\ge \underbar{Z}$. There are four cases. 
        \begin{enumerate}
            \item $B+\beta(r_1)<0$ and $B+\beta(r_2)\ge 0$. In this case, because $\beta$ is convex, $\bar{Z}' = B+\beta(r_1)+\beta(r_4) \ge B+\beta(r_2)+\beta(r_3) = \bar{Z}$ and $\underbar{Z}' = B+\beta(r_2)+\beta(r_3) > \omega = \underbar{Z}$.
            \item $B+\beta(r_2)<0$ and $B+\beta(r_3)\ge 0$. In this case, because $\beta$ is convex, $\bar{Z}' = B+\beta(r_1)+\beta(r_4) \ge B+\beta(r_2)+\beta(r_3) = \bar{Z}$ and $\underbar{Z}' = \omega = \underbar{Z}$.
            \item $B+\beta(r_3)<0$ and $B+\beta(r_4)\ge 0$. In this case, $\bar{Z}' = B+\beta(r_1)+\beta(r_4) \ge \omega = \bar{Z}$ and $\underbar{Z}' = \omega = \underbar{Z}$.
            \item $B+\beta(r_4)<0$. In this case, all the review processes terminate in round $i$, satisfying the monotonicity in ordering.
        \end{enumerate}
    \end{enumerate}
\end{enumerate}
\end{proof}

\section{Proofs in Section \ref{sec:endogeneous}}

\subsection{Marginal Rate of Substitution in the Binary Effort Setting}
\label{app:mrs-binary}

\begin{proof}[Proofs in \cref{lem:mrs-binary}]
    We first write down the expected utility and derive the MRS of the parallel review mechanism. 
    \begin{align*}
        U_a^p(n_h, n_l) &= u_h\sum_{k\in [n_h]} k\cdot\tbinom{n_h}{k} p_h^{k}(1-p_h)^{n_h-k} + u_l\sum_{k\in [n_l]} k\cdot\tbinom{n_l}{k} p_l^{k}(1-p_l)^{n_l-k}.
    \end{align*}
    Because under the parallel review mechanism, every paper is independently reviewed, the utility gain of generating one more paper it exactly equal to the expected utility of having that paper reviewed, i.e.~the product of acceptance probability and the reward of acceptance.
    Therefore,
    \begin{align*}
      MRS^p(n_h, n_l) = \frac{U_a^p(n_h+1, n_l) - U_a^p(n_h, n_l)}{U_a^p(n_h, n_l+1) - U_a^p(n_h, n_l)} = \frac{p_h u_h}{p_l u_l}>1.
    \end{align*}

    For the naive sequential review mechanism, note that the author will always rank high-quality papers before low-quality papers under $\pi^*$. Therefore,
    \begin{align*}
        U_a^s(n_h, n_l) &= u_h\sum_{k_h\in [n_h]} p_h^{k_h} + u_l\cdot p_h^{n_h}\sum_{k_l\in [n_l]} p_l^{k_l}
    \end{align*}
    The marginal return of writing a high-quality paper and a low-quality, respectively, can be written as
    \begin{align}
        U_a^s(n_h+1, n_l) - U_a^s(n_h, n_l) &= p_h^{n_h} u_h - p_h^{n_h}(1-p_h)u_l \sum_{k_l\in [n_l]} p_l^{k_l} = p_h^{n_h} \left(p_hu_h-\frac{1-p_h}{1-p_l}p_l(1-p_l^{n_l})u_l\right) \label{eq:margianl_h}\\
        U_a^s(n_h, n_l+1) - U_a^s(n_h, n_l) &= p_h^{n_h}p_l^{n_l+1}u_l. \label{eq:margianl_l}
    \end{align}
    Therefore,
    \begin{align*}
      MRS^s(n_h, n_l) &= \frac{p_hu_h-\frac{1-p_h}{1-p_l}p_l(1-p_l^{n_l})u_l}{p_l^{n_l+1}u_l} \\
                      &= \frac{p_h u_h}{p_l u_l}\cdot\frac{1-\frac{(1-p_h)p_l}{(1-p_l)p_h}\cdot\frac{u_l}{u_h}\cdot(1-p_l^{n_l})}{p_l^{n_l}}\\
                      &= MRS^p(n_h, n_l) \cdot \frac{1-\frac{(1-p_h)p_l}{(1-p_l)p_h}\frac{u_l}{u_h}\cdot(1-p_l^{n_l})}{p_l^{n_l}}\coloneqq MRS^p(n_h, n_l) \cdot \eta(n_l).
    \end{align*}

    We want to show $\eta(n_l)\ge 1$ for any $n_l\ge 0$. First note that because $0\le p_l<p_h\le 1$, $\frac{(1-p_h)p_l}{(1-p_l)p_h}< 1$. Furthermore, by assumption, $\frac{u_l}{u_h}\le 1$. Therefore, $\eta(n_l)\ge \frac{1-(1-p_l^{n_l})}{p_l^{n_l}} = 1$, and the inequality is strict when $n_l\ge 1$. This completes the proof.
\end{proof}

\subsection{Proof of Proposition \ref{thm:qual_vs_quant}}

\begin{proof}
    The proof follows by rewriting the expected reward $U_a(n'_h, n_l)$ using $U_a(n_h, n_l)$ and $U_a(n_h+1, n_l)$, and rewriting $U_a(n_h, n'_l)$ using $U_a(n_h, n_l)$ and $U_a(n_h, n_l+1)$. Then, we can prove the theorem using the property of the marginal rate of substitution shown in \cref{lem:mrs-binary}.

    \textbf{Parallel review mechanism. }
    We first rewrite the expected reward for the parallel review mechanism. Note that under the parallel review mechanism, the marginal reward of writing one paper depends only on the quality of that paper. That is, $U_a^p(n_h+1, n_l) - U_a^p(n_h, n_l) = p_hu_h$ for any $n_h, n_l$. Therefore
    \begin{align*}
        U_a^p(n'_h, n_l) &= U_a^p(n'_h, n_l) - U_a^p(n'_h-1, n_l) +  U_a^p(n'_h-1, n_l) - \cdots - U_a^p(n_h, n_l) + U_a^p(n_h, n_l)\\
        &= (n'_h - n_h)\cdot(U_a^p(n_h+1, n_l) - U_a^p(n_h, n_l)) + U_a^p(n_h, n_l).\\
    \intertext{Similarly, we have}
        U_a^p(n_h, n'_l) &= (n'_l - n_l)\cdot(U_a^p(n_h, n_l+1) - U_a^p(n_h, n_l)) + U_a^p(n_h, n_l).
    \end{align*}
    Therefore, the author's preference $U_a^p(n'_h, n_l)\ge U_a^p(n_h, n'_l)$ implies that 
    \begin{align}
        &\qquad U_a^p(n'_h, n_l) - U_a^p(n_h, n'_l) \ge 0\notag\\
        \Leftrightarrow&\qquad (n'_h - n_h)\cdot(U_a^p(n_h+1, n_l) - U_a^p(n_h, n_l)) - (n'_l - n_l)\cdot(U_a^p(n_h, n_l+1) - U_a^p(n_h, n_l)) \ge 0\notag\\
        \Leftrightarrow&\qquad MRS^p_{h,l}(n_h,n_l)\ge \frac{n'_l - n_l}{n'_h - n_h}\label{eq:pal_cond}
    \end{align}

    \textbf{Sequential review mechanism. }
    The analysis for the sequential review mechanism is a little more complicated. For simplicity, let $\lambda_h(k) = \frac{U_a^s(n_h+k, n_l) - U_a^s(n_h+k-1, n_l)}{U_a^s(n_h+1, n_l) - U_a^s(n_h, n_l)}$ for any $k\ge 1$. Analogously, let $\lambda_l(k) = \frac{U_a^s(n_h, n_l+k) - U_a^s(n_h, n_l+k-1)}{U_a^s(n_h, n_l+1) - U_a^s(n_h, n_l)}$. Note that both $\lambda_h(k)$ and $\lambda_l(k)$ are positive for any $k$. Using the same recipe as the analysis of the parallel review mechanism,
    \begin{align*}
        U_a^s(n'_h, n_l) &= U_a^s(n'_h, n_l) - U_a^s(n'_h-1, n_l) + U_a^s(n'_h-1, n_l) - \cdots - U_a^s(n_h, n_l) + U_a^s(n_h, n_l)\\
        &= (1 + \lambda_h(2)+\cdots+ \lambda_h(n'_h-n_h))\cdot(U_a^s(n_h+1, n_l) - U_a^s(n_h, n_l)) + U_a^s(n_h, n_l).\\
    \intertext{Similarly,}
    U_a^s(n_h, n'_l) &= (1 + \lambda_l(2)+\cdots+ \lambda_l(n'_l-n_l))\cdot(U_a^s(n_h, n_l+1) - U_a^s(n_h, n_l)) + U_a^s(n_h, n_l).
    \end{align*}

     The following lemma simplifies $\lambda_h$ and $\lambda_l$.
    \begin{lemma}\label{lem:lambda}
        $\lambda_h(k) = p_h^{k-1}$ and $\lambda_l(k) = p_l^{k-1}$.
    \end{lemma}

    We want to show that whenever \cref{eq:pal_cond} holds, the following inequality holds.
    \begin{align}
        &\qquad U_a^s(n'_h, n_l) - U_a^s(n_h, n'_l) \ge 0\notag\\
        \Leftrightarrow&\qquad\sum_{k=1}^{n'_h-n_h}\lambda_h(k)\cdot(U_a^s(n_h+1, n_l) - U_a^s(n_h, n_l)) - \sum_{k=1}^{n'_l-n_l}\lambda_l(k)\cdot(U_a^s(n_h, n_l+1) - U_a^s(n_h, n_l)) \ge 0\notag\\
        \Leftrightarrow&\qquad MRS^s_{h,l}(n_h,n_l)\ge \frac{\sum_{k=1}^{n'_l-n_l}\lambda_l(k)}{\sum_{k=1}^{n'_h-n_h}\lambda_h(k)}\notag\\
        \Leftrightarrow&\qquad MRS^s_{h,l}(n_h,n_l)\ge \frac{\sum_{k=1}^{n'_l-n_l}p_l^{k-1}}{\sum_{k=1}^{n'_h-n_h}p_h^{k-1}}.\label{eq:seq_cond}
    \end{align}

    First note that \cref{eq:seq_cond} trivially holds when $n'_h-n_h\ge n'_l-n_l$. This is because by \cref{lem:mrs-binary}, the left-hand-side of \cref{eq:seq_cond} is greater than 1, while because $p_h>p_l$, the right-hand-side of \cref{eq:seq_cond} is smaller than 1. Therefore, we can focus on the case where $n'_h-n_h< n'_l-n_l$.

    By \cref{lem:mrs-binary}, we know that $MRS^s_{h,l}(n_h,n_l)\ge MRS^p_{h,l}(n_h,n_l)$. Therefore, to show that \cref{eq:seq_cond} holds whenever \cref{eq:pal_cond} holds, it is sufficient to show that 
    
    \begin{align}
        \qquad\frac{n'_l - n_l}{n'_h - n_h} &\ge \frac{\sum_{k=1}^{n'_l-n_l}p_l^{k-1}}{\sum_{k=1}^{n'_h-n_h}p_h^{k-1}}\notag\\
        \Leftrightarrow\qquad \frac{\sum_{k=1}^{n'_h-n_h}p_h^{k-1}}{n'_h - n_h} &\ge \frac{\sum_{k=1}^{n'_l-n_l}p_l^{k-1}}{n'_l - n_l}.\label{eq:goal_cond}
    \intertext{\Cref{eq:goal_cond} holds because }
        \frac{\sum_{k=1}^{n'_h-n_h}p_h^{k-1}}{n'_h - n_h} &\ge \frac{\sum_{k=1}^{n'_h-n_h}p_l^{k-1}}{n'_h - n_h} \tag{because $p_h>p_l$}\\
        &\ge \frac{\sum_{k=1}^{n'_l-n_l}p_l^{k-1}}{n'_l - n_l} \tag{because $n'_h-n_h < n'_l-n_l$}.
    \end{align}
    
    This completes the proof of the first part of the theorem.

    To prove the second part of the theorem, we only have to find a counter example. Suppose $p_h = 1$, $p_l = 0.5$, and let $n_h = n_l = 0$, $n_h' = 1$ and $n_l' = 3$. In this example, $U_a^s(n_h', n_l) = 1$ and $U_a^s(n_h, n_l') = 0.5+0.5^2+0.5^3 = 0.875$, implying that $U_a^s(n_h', n_l) > U_a^s(n_h, n_l')$. However, $U_a^p(n_h', n_l) = 1$ and $U_a^p(n_h, n_l') = 1.5$, implying that $U_a^p(n_h', n_l) < U_a^s(n_h, n_l')$. This completes the proof of \cref{thm:qual_vs_quant}.
\end{proof}

Now, we prove \cref{lem:lambda}.
\begin{proof}[Proof of \cref{lem:lambda}] 
We present the proof for $\lambda_h$ while the proof for $\lambda_l$ is analogous. 
By definition
    \begin{align}
        \lambda_h(k) &= \frac{U_a^s(n_h+k, n_l) - U_a^s(n_h+k-1, n_l)}{U_a^s(n_h+1, n_l) - U_a^s(n_h, n_l)}\notag\\
        &= \frac{p_h^{n_h+k-1} \left(p_hu_h-\frac{1-p_h}{1-p_l}p_l(1-p_l^{n_l})u_l\right)}{p_h^{n_h} \left(p_hu_h-\frac{1-p_h}{1-p_l}p_l(1-p_l^{n_l})u_l\right)}\tag{by \cref{eq:margianl_h}}\\
        &= p_h^{k-1}.
    \end{align}
    \end{proof}

\subsection{Generalizations to the Finite Effort Setting}
\label{app:finite_effort}

We first generalize the key property of marginal rate of substitution to the finite effort setting

\begin{proof}[Proof of \cref{lem:mrs-finite}]
    The proof of this proposition is analogue to the proof of \cref{lem:mrs-binary}. For simplicity, let 
    $$\Gamma_i = \prod_{k\in [i]} p_k^{n_k}, \qquad S_i = \sum_{k\in [n_i]} p_i^k = \frac{p_i(1-p_i^{n_i})}{1-p_i}.$$

    Again, because each paper is independently reviewed under the parallel review mechanism, it is easy to show 
    \begin{equation*}
        MRS^p_{i,j}(\bm{n}) = \frac{p_i u_i}{p_j u_j}.
    \end{equation*}

    The MRS of the sequential review mechanism is more complicated. We first write down the author's expected reward of writing $\bm{n}$ papers and rank them truthfully.
    \begin{align*}
        U_a^s(\bm{n}) &= \sum_{l\in [|\bm{n}|]} \Pr(\text{paper $1,\ldots, l$ are accepted})\cdot \text{$<$the reward of paper $l>$}\\
        &= S_1u_1 +\Gamma_1 S_2u_2 +\cdots + \Gamma_{m-1} S_m u_m
    \end{align*}

    For simplicity, let $Z_i = \sum_{l=i}^m \Gamma_{l-1} S_l u_l$ for any $2\le i \le m$ and let $Z_1 = S_1 u_1 +  \sum_{l=2}^m \Gamma_{l-1} S_l u_l$. Thus, $U_a^s(\bm{n}) = Z_1$. 
    Now, suppose the author writes one more paper with acceptance probability $p_i$. Recall that we use $\bm{n}' = (n_1,\ldots, n_{i-1}, n_i+1, \ldots, n_m)$ to denote the new vector of number of papers.
    \begin{align*}
        U_a^s(\bm{n}') &= S_1u_1 +\Gamma_1 S_2u_2 +\cdots + \Gamma_{i-1} S_i u_i + \Gamma_{i} p_i u_i + p_i\left(\Gamma_{i} S_{i+1} u_{i+1} + \cdots +  \Gamma_{m-1} S_m u_m\right)\\
        &= Z_1 - Z_{i+1} + \Gamma_{i} p_i u_i + p_i Z_{i+1}\\
        &=  Z_1 - (1-p_i) Z_{i+1} + \Gamma_{i} p_i u_i.
    \end{align*}
    Therefore, 
    \begin{align}\label{eq:marginal_i}
        U_a^s(\bm{n}') - U_a^s(\bm{n}) = \Gamma_{i} p_i u_i - (1-p_i) Z_{i+1}.
    \end{align}

    Similarly, if the author writes one more paper with acceptance probability $p_j$ where $j>i$, we have
    \begin{align*}
        U_a^s(\bm{n}'') - U_a^s(\bm{n}) &= \Gamma_{j} p_j u_j - (1-p_j) Z_{j+1}.
    \end{align*}

    Therefore, 
    \begin{align*}
        MRS^s_{i,j}(\bm{n}) &= \frac{U_a^s(\bm{n}') - U_a^s(\bm{n})}{U_a^s(\bm{n}'') - U_a^s(\bm{n})}\\
        &= \frac{\Gamma_{i} p_i u_i - (1-p_i) Z_{i+1}}{\Gamma_{j} p_j u_j - (1-p_j) Z_{j+1}}\\
        &= \frac{p_iu_i}{p_ju_j}\cdot \frac{\Gamma_{i} - \frac{1-p_i}{p_iu_i} Z_{i+1}}{\Gamma_{j} - \frac{1-p_j}{p_ju_j} Z_{j+1}}
        \coloneqq MRS^p_{i,j}(\bm{n}) \cdot \eta(\bm{n}).
    \end{align*}
    We want to show $\eta(\bm{n}) - 1 \ge 0$. Because $U_a^s(\bm{n}'') - U_a^s(\bm{n})>0$, it is sufficient to show
    \begin{align*}
        h(\bm{n}) = \Gamma_i -\Gamma_j + \frac{1-p_j}{p_ju_j} Z_{j+1} - \frac{1-p_i}{p_iu_i} Z_{i+1} \ge 0.
    \end{align*}
    This inequality holds because 
    \begin{align*}
        h(\bm{n}) &= \Gamma_i -\Gamma_j + \frac{1-p_j}{p_ju_j} Z_{j+1} - \frac{1-p_i}{p_iu_i} (Z_{i+1} - Z_{j+1}) - \frac{1-p_i}{p_iu_i} Z_{j+1} \\
        \intertext{Because $p_i>p_j$ and $u_i\ge u_j$, $\frac{1-p_i}{p_iu_i}<\frac{1-p_j}{p_ju_j}$.}
        &\ge \Gamma_i -\Gamma_j - \frac{1-p_i}{p_iu_i} (Z_{i+1} - Z_{j+1}). \tag{The inequality is strict if $Z_{j+1}>0$, i.e.~$\sum_{k = j+1}^m n_k > 0$.}\\
    \end{align*}
    Note that 
    \begin{align*}
        \frac{1-p_i}{p_iu_i} Z_k &= \frac{1-p_i}{p_iu_i} \cdot \sum_{l=k}^m \Gamma_{l-1} S_l u_l = \sum_{l=k}^m \frac{1-p_i}{p_i} \cdot \frac{p_l(1-p_l^{n_l})}{1-p_l} \cdot \frac{u_l}{u_i}\cdot \Gamma_{l-1}. 
    \intertext{Therefore, for any $k>i$, we have }
    \frac{1-p_i}{p_iu_i} Z_k&\le \sum_{l=k}^m (1-p_l^{n_l}) \Gamma_{l-1} = \sum_{l=k}^m (\Gamma_{l-1} - \Gamma_l) = \Gamma_{k-1} - \Gamma_m.
    \end{align*}
    Furthermore, the above inequality is strict if $\sum_{l = k}^m n_k > 0$.
    Now, we complete the proof by showing $h(\bm{n})$ is non-negative.
    \begin{align*}
        h(\bm{n}) &\ge \Gamma_i -\Gamma_j - \frac{1-p_i}{p_iu_i} (Z_{i+1} - Z_{j+1})\\
        & \ge \Gamma_i -\Gamma_j - ( \Gamma_i - \Gamma_m -  \Gamma_j + \Gamma_m)\\
        & = 0.
    \end{align*}
    If $\sum_{k = i+1}^m n_k > 0$, at least one of the above two inequality is strict, and thus $MSR^s_{i,j}(\bm{n}) > MSR^p_{i,j}(\bm{n})$ for any $\bm{n}$ and $i<j$.
\end{proof}

Then, we generalize \cref{lem:lambda} to the finite effort setting respectively. If \cref{lem:lambda} can be generalized, the proof of \cref{thm:qual_vs_quant} straightforwardly generalizes because we can treat every high-quality paper in the binary setting as a paper with acceptance probability $p_i$, and every low-quality paper as a paper with acceptance probability $p_j$ for any $i<j$. Let $\lambda_i(k)=\frac{U_a^s(\bm{n}_i^{(k)}) - U_a^s(\bm{n}_i^{(k-1)})}{U_a^s(\bm{n}_i^{(1)}) - U_a^s(\bm{n}_i^{(0)})}$, where $\bm{n}_i^{(k)} = (n_1, \ldots, n_{i-1}, n_i + k, \dots, n_m)$.

    \begin{lemma}\label{lem:lambda_finite}
        $\lambda_i(k)=p_i^{k-1}$.
    \end{lemma}

    \begin{proof}
        For simplicity, let 
    $$\Gamma_i = \prod_{k\in [i]} p_k^{n_k}, \qquad S_i = \sum_{k\in [n_i]} p_i^k = \frac{p_i(1-p_i^{n_i})}{1-p_i}, \qquad Z_i = \sum_{l=i}^m \Gamma_{l-1} S_l u_l.$$

    By \cref{eq:marginal_i}, we know that if there is one more paper with acceptance probability $p_i$, the gain of expected reward is $U_a^s(\bm{n}_i^{(1)}) - U_a^s(\bm{n}_i^{(0)}) = \Gamma_{i} p_i u_i - (1-p_i) Z_{i+1}$. Intuitively, if there are $k$ more papers of acceptance probability $p_i$, while reasoning about the expected reward, it is equivalent to say that there is a new category of papers that have the same acceptance probability $p_i$ and are ranked lower than the $n_{i-1}$ papers with $p_{i-1}$ but higher than the $n_i$ papers with $p_i$. Therefore, every term of the gain of expected reward is exponentially discounted, i.e.~$U_a^s(\bm{n}_i^{(k)}) - U_a^s(\bm{n}_i^{(k-1)}) = p_i^{k-1}\cdot \left(U_a^s(\bm{n}_i^{(1)}) - U_a^s(\bm{n}_i^{(0)})\right)$. This implies that $\lambda_i(k) = p_i^{k-1}$.
    \end{proof}

\end{document}